\documentclass[oneside,english,12pt]{iopart}
\usepackage[T1]{fontenc}
\usepackage[latin9]{inputenc}
\usepackage{float}
\usepackage{units}
\usepackage{mathrsfs}
\usepackage{url}
\usepackage{amsbsy}
\usepackage{amstext}
\usepackage{amssymb}
\usepackage{graphicx}
\usepackage{esint}

\makeatletter

\floatstyle{ruled}
\newfloat{algorithm}{tbp}{loa}
\providecommand{\algorithmname}{Algorithm}
\floatname{algorithm}{\protect\algorithmname}

\usepackage{algorithmic} 
\usepackage{algorithm} 
\usepackage{hyphenat}
\usepackage{animate}
\usepackage{amsthm}

\newtheorem{theorem}{Theorem}
\newcommand{\eqref}[1]{(\ref{#1})}

\makeatother

\usepackage{babel}

\begin{document}
\title{The Discrete Fourier Transform for Golden Angle Linogram Sampling}
\author{Elias S. Helou$^1$, Marcelo V. W. Zibetti$^2$, Leon Axel$^2$,\\ Kai Tobias Block$^2$, Ravinder R. Regatte$^2$ and\\ Gabor T. Herman$^3$}
\address{$^1$Institute of Mathematical Sciences and Computation, University of São Paulo, São Carlos, SP, Brazil}
\address{$^2$Center for Biomedical Imaging, Department of Radiology, New York University School of Medicine, New York, NY, USA}
\address{$^3$Ph.D. Program in Computer Science, City University of New York, New York, NY, USA}
\eads{\mailto{elias@icmc.usp.br}, \mailto{Marcelo.WustZibetti@nyulangone.org}, \mailto{Leon.Axel@nyulangone.org}, \mailto{KaiTobias.Block@nyulangone.org}, \mailto{Ravinder.Regatte@nyulangone.org} and \mailto{gabortherman@yahoo.com}}
\begin{abstract}
Estimation of the Discrete-Time Fourier Transform (DTFT) at points
of a finite domain arises in many imaging applications. A new approach
to this task, the Golden Angle Linogram Fourier Domain (GALFD), is
presented, together with a computationally fast and accurate tool,
named Golden Angle Linogram Evaluation (GALE), for approximating the
DTFT at points of a GALFD. A GALFD resembles a Linogram Fourier Domain
(LFD), which is efficient and accurate. A limitation of linograms
is that embedding an LFD into a larger one requires many extra points,
at least doubling the domain's cardinality. The GALFD, on the other
hand, allows for incremental inclusion of relatively few data points.
Approximation error bounds and floating point operations counts are
presented to show that GALE computes accurately and efficiently the
DTFT at the points of a GALFD. The ability to extend the data collection
in small increments is beneficial in applications such as Magnetic
Resonance Imaging. Experiments for simulated and for real-world data
are presented to substantiate the theoretical claims. The mathematical
analysis, algorithms, and software developed in the paper are equally
suitable to other angular distributions of rays and therefore we bring
the benefits of linograms to arbitrary radial patterns.
\end{abstract}

\submitto{\IP}
\maketitle

\section{Introduction}

We use $\mathbb{R}$, $\mathbb{C}$ and $\mathbb{Z}$ to denote the
sets of all real numbers, complex numbers and integers, respectively.
Let $\boldsymbol{x}\in\mathbb{C}^{mn}$ denote a fixed two-dimensional
complex array with $m$ rows and $n$ columns and let $\imath:=\sqrt{-1}$.
We define the Discrete-Time Fourier Transform (DTFT) $\mathcal{D}[\boldsymbol{x}]:\mathbb{R}^{2}\to\mathbb{C}$
by
\begin{equation}
\mathcal{D}[\boldsymbol{x}](\xi,\upsilon):=\sum_{i=0}^{m-1}\sum_{j=0}^{n-1}x_{i,j}e^{-\imath(j\xi+i\upsilon)},\label{eq:DSFT_def}
\end{equation}
for $(\xi,\upsilon)\in\mathbb{R}^{2}$. Given a domain $\mathfrak{D}\subset\mathbb{R}^{2}$
with $|\mathfrak{D}|=MN$, where $M$ and $N$ are positive integers,
the Discrete Fourier Transform (DFT) over the domain $\mathfrak{D}$
is the linear operator $\mathrm{D}_{\mathfrak{D}}:\mathbb{C}^{mn}\to\mathbb{C}^{MN}$
that takes any $\boldsymbol{x}\in\mathbb{C}^{mn}$ as input and outputs
the values of the DTFT $\mathcal{D}[\boldsymbol{x}](\xi,\upsilon)$
for all $(\xi,\upsilon)\in\mathfrak{D}$. The definition of the DFT's
domain $\mathfrak{D}$ plays an important role, influencing both the
ability of the DFT to model practical problems and the possibility
of fast computation of $\mathrm{D}_{\mathfrak{D}}\boldsymbol{x}$.
Because the DTFT is $2\pi$-periodic, we can assume that $\mathfrak{D}\subset[-\pi,\pi]^{2}$.

In this paper we introduce a new family $\mathscr{G}\subset2^{\mathbb{R}^{2}}$
(for a given set $\mathbb{S}$, $2^{\mathbb{S}}$ denotes the power
set of $\mathbb{S}$, that is, the set of all subsets of $\mathbb{S}$)
of domains that is such that a ``fine-grained inclusion chain''
of the domains of $\mathscr{G}$ exists. By this we mean that, for
any domain $\mathfrak{D}_{1}\in\mathscr{G}$, there exists another
domain $\mathfrak{D}_{2}\in\mathscr{G}$ such that $\mathfrak{D}_{1}\subset\mathfrak{D}_{2}$
and that $\mathfrak{D}_{2}-\mathfrak{D}_{1}$ has few points compared
to the cardinality of $\mathfrak{D}_{1}$. This latter property is
one of the things setting this work apart from other existing techniques;
in particular, it creates practical interest in certain Magnetic Resonance
Imaging (MRI) applications.\footnote{We provide a few words to clarify the relationship between \eqref{eq:DSFT_def}
and its use in solving some inverse problems, such as MRI. In the
common parlance of Inverse Problems, the $\boldsymbol{x}$ in \eqref{eq:DSFT_def}
comprises the model parameters that we wish to recover from the observed
data, indicated by the left-hand-side of \eqref{eq:DSFT_def}. In
our discussion the forward operator is not fixed a priori; we may
make it any $\mathrm{D}_{\mathfrak{D}}$, with $\mathfrak{D}\subset\mathbb{R}^{2}$
and $|\mathfrak{D}|=MN$. Our concern is the efficacious choice of
$\mathscr{G}$ that contains the $\mathfrak{D}$s to be used in a
particular application.}

In addition to the family $\mathscr{G}$ being capable of modeling
practical data acquisition modes, it is also desirable from the numerical-analysis
point of view: For every domain $\mathfrak{D}\in\mathscr{G}$, there
is an efficient computational approach that delivers an accurate approximate
evaluation of $\mathrm{D}_{\mathfrak{D}}\boldsymbol{x}$ for arbitrary
data $\boldsymbol{x}\in\mathbb{C}^{mn}$. A contribution of the present
paper is a numerical scheme for the approximate computation of $\mathrm{D}_{\mathfrak{D}}\boldsymbol{x}$
for any domain $\mathfrak{D}\in\mathscr{G}$ and also (this is important
for reconstruction techniques) for the approximate computation of
the adjoint operator.\footnote{We assume that, for any of our operators O, the meaning of its adjoint
O{*} is understood.} We include a bound for the approximation error, as well as a floating
point operations (flops) count in order to give theoretical support
to our accuracy and efficiency claims. These claims are also confirmed
by simulation experiments and in MRI reconstructions from real-world
data.

As will be seen, the key features of the method that we introduce
are the following:\footnote{We developed an efficient open source C++ implementation featuring
the listed characteristics. MATLAB and Python bindings are also available.
The source code is at \url{https://bitbucket.org/eshneto/gale/downloads/}.}
\begin{itemize}
\item It performs fast computation of the DFT for the proposed sampling.
\item Increased accuracy can be achieved without considerably increasing
computation time.
\item The method is very accurate in the low-frequency part of the spectrum,
which is the most important part in many imaging applications.
\item Adjoint computation is as fast as the forward DFT.
\item An upper bound for the approximation error is available.
\item It approximates by sum truncation that converges quickly with more
terms.
\item The computational cost of the method scales linearly with the length
of the truncated sum; other known methods scale with the square of
this length\footnote{We discuss some of these methods in Section~\ref{sec:Numerical-Experimentation}.
We do not introduce these techniques earlier in order to avoid disrupting
the presentation.}.
\item Parallelization of the method is efficient and its memory access pattern
is cache-friendly.
\end{itemize}
\par{}Other families of domains, usually sampled in a polar pattern,
that allowed for inclusion of a relatively small amount of data have
been known~\cite{wsk06} and used~\cite{bth18,cld17,fgb14} in MRI
for a while. However, linograms have never been brought into this
context and the idea of joining golden angle (or more general) angular
distributions with linogram-based~\cite{edh87,ehr88,ahr90} sampling
is new to the present work. Not only it results in computationally
effective methods, as we extensively discuss in this paper, but the
wider coverage of the Fourier space of linograms compared to polar
domains has recently been shown to bring improvements in image quality
as well~\cite{erm16,yll17}. Our contribution here is, therefore,
twofold. We introduce the new family of golden-angle linogram domains
and provide an effective and accurate method for the computation of
the DTFT at the points of these domains\footnote{The method we propose is capable of handling arbitrary angular linogram
sampling, we focus on golden angle for concreteness.}. In what follows next, we give, as examples, some classical families
of domains. We delay till \ref{sub:Motivate} discussions of how various
domains have been used in practical MRI and why we believe that the
approaches demonstrated in the present work will lead to improved
performance of MRI in certain kinds of medical applications.

\subsection{The Cartesian and Polar Fourier Domains}

The Cartesian Fourier Domain (CFD) $\mathfrak{C}_{M,N}$, for $M$
and $N$ positive integers, is defined by
\begin{eqnarray*}
\fl\mathfrak{C}_{M,N}:=\Biggl\{\left(\frac{2\pi I}{M},\frac{2\pi J}{N}\right):I\in\{-\nicefrac{M}{2},-\nicefrac{M}{2} & {}+1,\dots,\nicefrac{M}{2}-1\},\\
 & \text{and}\quad J\in\left\{ -\nicefrac{N}{2},-\nicefrac{N}{2}+1,\dots,\nicefrac{N}{2}-1\right\} \Biggr\},
\end{eqnarray*}
see Figure \ref{fig:domains}, left. If $M\ge m$ and $N\ge n$, there
are Fast Fourier Transform (FFT)~\cite{cot65} algorithms that can
compute $\mathcal{D}[\boldsymbol{x}](\xi,\upsilon)$ for all points
$(\xi,\upsilon)\in\mathfrak{C}_{M,N}$ very efficiently, especially
when $M$ and $N$ are powers of $2$. Computing $\mathrm{D}_{\mathfrak{C}_{M,N}}\boldsymbol{x}$
directly from the definition requires $O(mnMN)$ flops, but FFTs perform
a mathematically equivalent computation that uses only $O(MN\log_{2}(MN))$
flops. 

\begin{figure}
\centering{}\includegraphics[width=0.5\columnwidth]{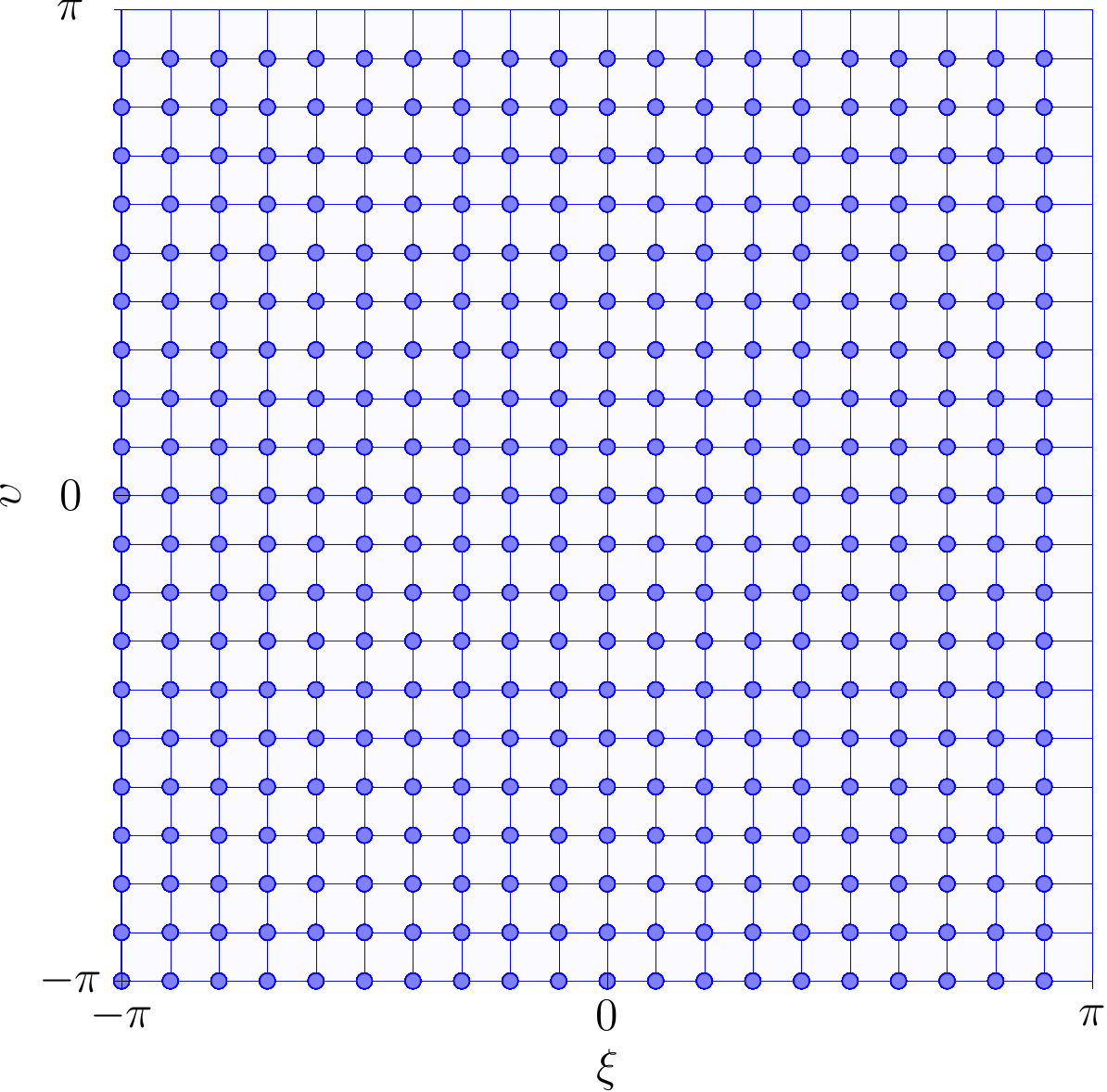}\includegraphics[width=0.5\columnwidth]{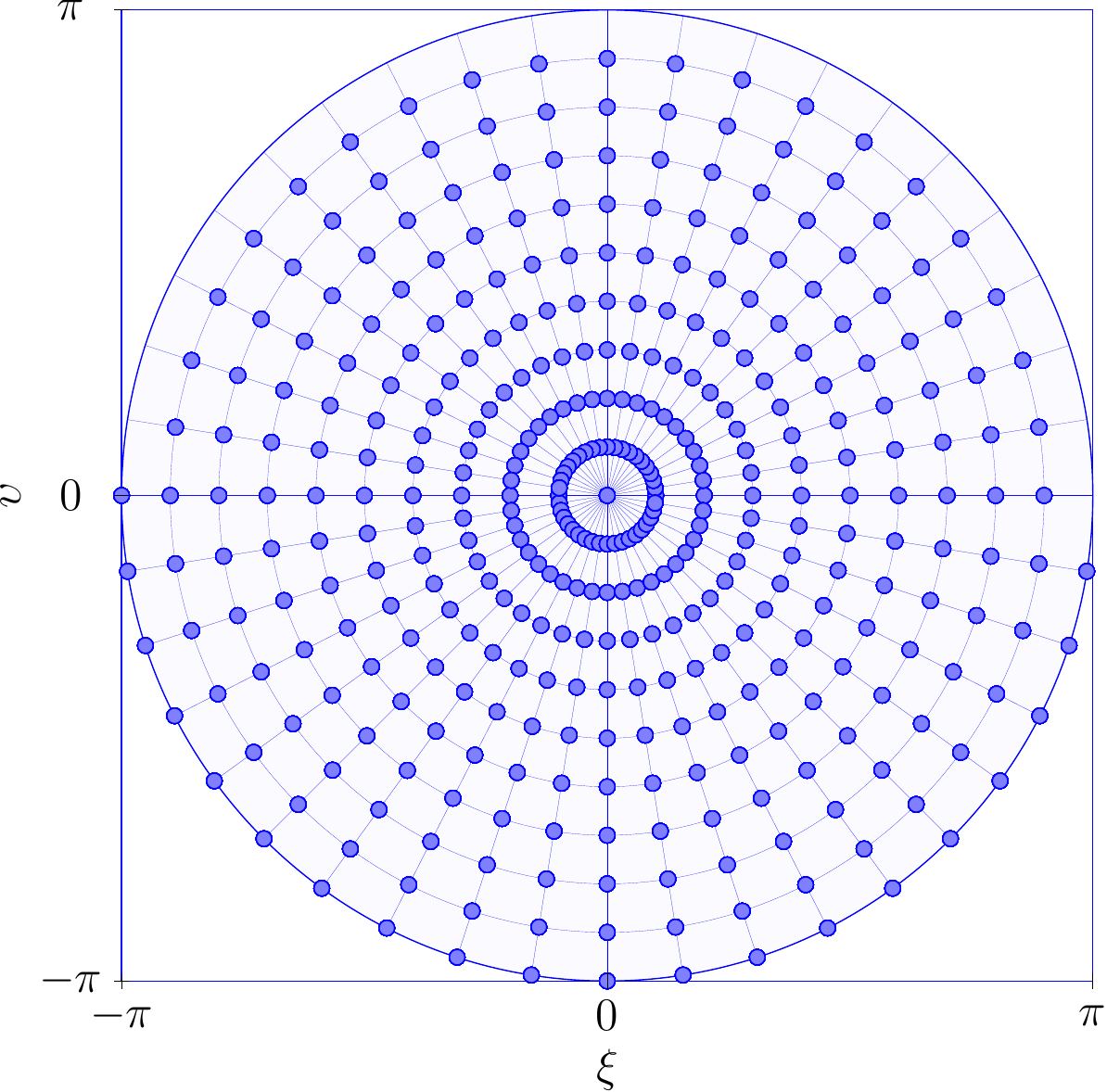}\caption{\label{fig:domains}The above images depict all the points contained
in two of the discrete domains described in the Introduction. Left:
Cartesian Fourier Domain $\mathfrak{C}_{20,20}$. Right: Polar Fourier
Domain $\mathfrak{P}_{20,20}$.}
\end{figure}

Note that, for $\mu$ and $\nu$ positive integers, $\mathfrak{C}_{M,N}\subset\mathfrak{C}_{\mu M,\nu N}$.
Thus, it is possible to add new points to a CFD, but this demands
at least doubling the DFT domain's cardinality, for otherwise the
samples of the DTFT computed at the points of $\mathfrak{C}_{M,N}$
will not be included in a new domain $\mathfrak{C}_{M',N'}$ unless
both $\nicefrac{M'}{M}$ and $\nicefrac{N'}{N}$ are positive integers.
This prevents incremental inclusion of a small amount of new sample
points since, for using a CFD of slightly larger cardinality, the
sample points of the smaller CFD will not be part of the larger domain.

Moreover, in some applications a radial sampling scheme of the Fourier
space may be desirable, such as in MRI~\cite{buf07,ysm05}, or unavoidable,
such as in Computerized Tomography (CT)~\cite{her09,nat86}, in which
case the efficient computation of the DTFT for such a radial set of
points is useful. Therefore, for these applications we may need to
consider the evaluation of the DFT over a Polar Fourier Domain (PFD)
$\mathfrak{P}_{M,N}$, defined by (see Figure \ref{fig:domains},
right)

\begin{eqnarray*}
\fl\mathfrak{P}_{M,N}:=\Biggl\{\Biggl(\frac{2\pi I}{M}\cos\frac{\pi J}{N},\frac{2\pi I}{M}\sin\frac{\pi J}{N}\Biggr):J & {}\in\{0,1,\dots,N-1\}.\\
 & \text{and}\quad I\in\left\{ -\nicefrac{M}{2},-\nicefrac{M}{2}+1,\dots,\nicefrac{M}{2}-1\right\} \Biggr\}.
\end{eqnarray*}
Reasonably fast and accurate techniques are available that can evaluate
$\mathcal{D}[\boldsymbol{x}]$ over a PFD \cite{fes03,kkp09}. For
a PFD, as was the case with the CFDs, no obvious way is available
for including a small number of extra points in the domain. Instead,
once more an inclusion of the form $\mathfrak{P}_{M,N}\subset\mathfrak{P}_{\mu M,\nu N}$
holds for positive integers $\mu$ and $\nu$, which is again of little
practical value in many applications.

In the remainder of this introduction, we present the details of the
newly proposed domains (the GALFDs), which have the flexibility of
the Golden Angle Polar Fourier Domains (GAPFDs) and the benefit of
the fast and accurate computation of the DFT over a Linogram Fourier
Domain (LFD). Both the GAPFD and the LFD are presented below to motivate
our choices, after which we introduce the new GALFD domains. Section~\ref{sec:Method}
provides the mathematical theory required for the development of the
technique, named GALE, that we propose for the approximate computation
of the DFT over various domains, including the GALFD. Section~\ref{sec:Algorithm}
discusses the details of turning the theory into a computational methodology.
Section~\ref{sec:Numerical-Experimentation} reports on our numerical
experiments using both simulated data and actual MRI data. Section~\ref{sec:Conclusions}
gives our conclusions. The appendices provide (for the sake of completeness)
some notations and algorithms that are either standard or well known
and also a discussion of the relevance of our work to practical MRI
in medicine.

\subsection{The Linogram Fourier Domain}

We define the subset $\mathfrak{L}_{M,N}$ of the real plane, called
an LFD, by $\mathfrak{L}_{M,N}:=\mathfrak{H}_{M,N}\cup\mathfrak{V}_{M,N}$,
where the sets $\mathfrak{H}_{M,N}$, and $\mathfrak{V}_{M,N}$ are
given according to
\begin{eqnarray*}
\fl\mathfrak{H}_{M,N}:=\Biggl\{\Biggl(\frac{2\pi I}{M},\frac{2\pi I}{M}\frac{4J}{N}\Biggr):J\in\{-\nicefrac{N}{4} & {}+1,-\nicefrac{N}{4}+2,\dots,\nicefrac{N}{4}\}\\
 & \text{and}\quad I\in\left\{ -\nicefrac{M}{2},-\nicefrac{M}{2}+1,\dots,\nicefrac{M}{2}-1\right\} \Biggr\},
\end{eqnarray*}
\begin{eqnarray*}
\fl\mathfrak{V}_{M,N}:=\Biggl\{\Biggl(\frac{2\pi I}{M}\frac{4J}{N},\frac{2\pi I}{M}\Biggr):J\in\{-\nicefrac{N}{4} & ,-\nicefrac{N}{4}+1,\dots,\nicefrac{N}{4}-1\}\\
 & \text{and}\quad I\in\left\{ -\nicefrac{M}{2}+1,-\nicefrac{M}{2}+2,\dots,\nicefrac{M}{2}\right\} \Biggr\}.
\end{eqnarray*}
(See \cite[(9.27), (9.28)]{her09}). An LFD is comprised of the intersections
of $M$ equally spaced concentric squares and $N$ rays going through
the origin, totaling $(M-1)N+1$ different points. Figure~\ref{fig:The-LSS},
left, shows an example of such a domain.\footnote{The name linogram was originally assigned to a sampling methodology
for CT reconstruction that allowed fast and precise inversion of projection
data \cite[Section 9.3]{edh87,ehr88,her09}. The underlying idea is
to use a parametrization $(u,v)$ of lines in the plane with the property
that the locus of all points in $(u,v)$ space that corresponds to
the set of lines that go through a fixed point in the plane is itself
a line; hence the name ``linogram.''} As before, it is not possible to add just a small number of points
to an LFD, even though for positive integers $\mu$ and $\nu$ we
have $\mathfrak{L}_{M,N}\subset\mathfrak{L}_{\mu M,\nu N}$.

\begin{figure}
\centering{}\includegraphics[width=0.5\columnwidth]{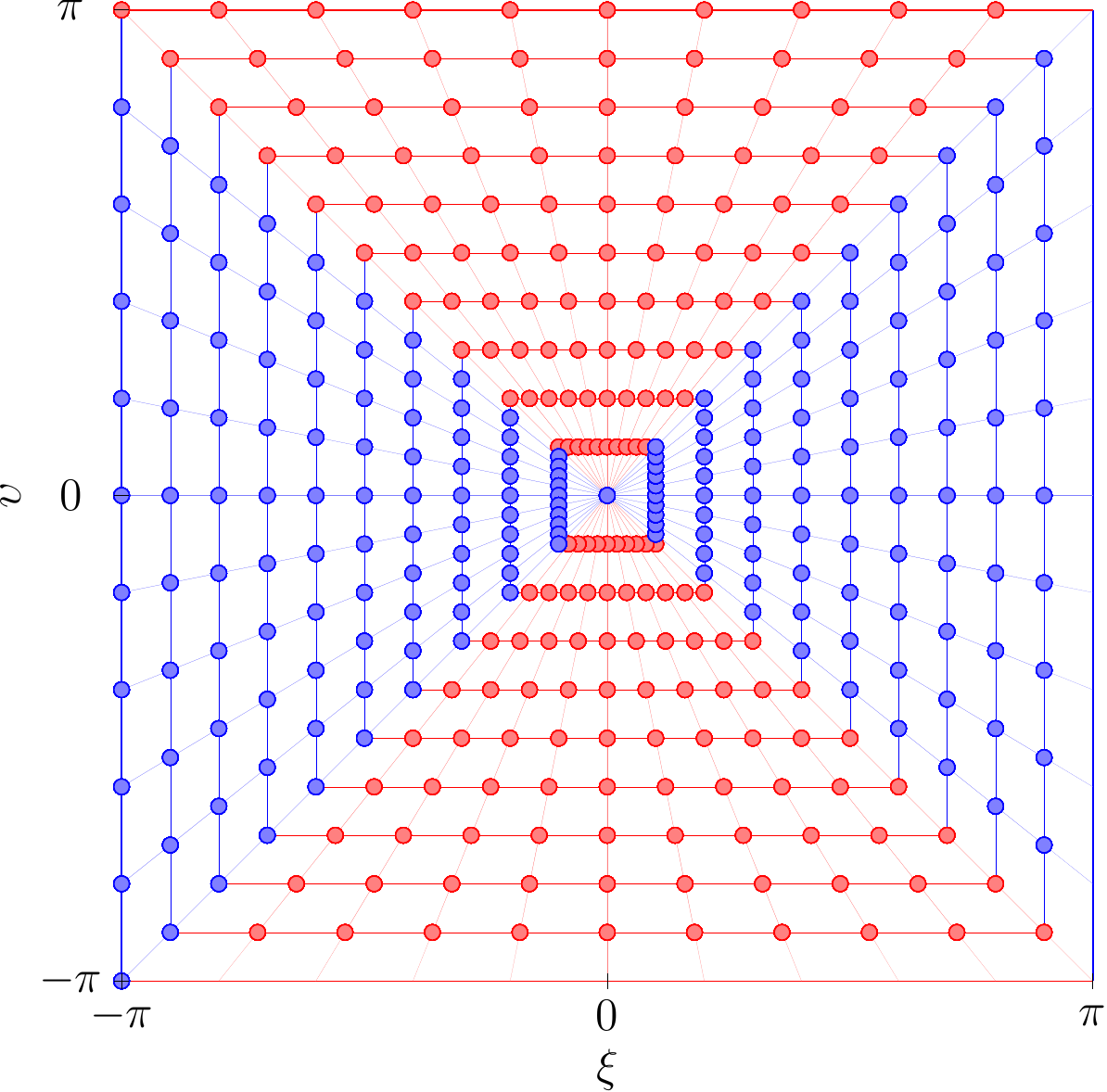}\includegraphics[width=0.5\columnwidth]{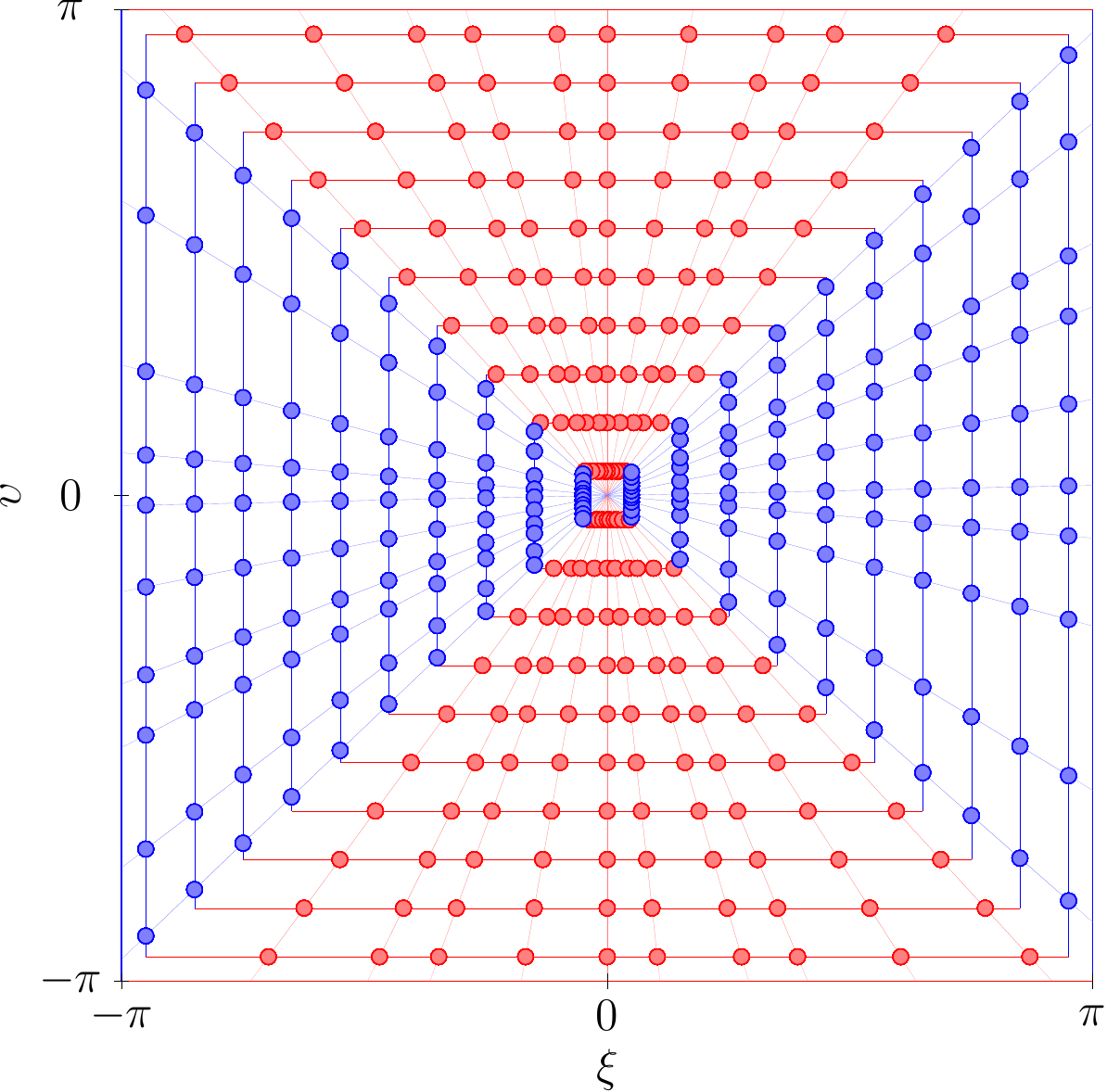}\caption{Left: The LFD $\mathfrak{L}_{20,20}$. The points of $\mathfrak{V}_{20,20}$
are displayed in red and the points of $\mathfrak{H}_{20,20}$ are
displayed in blue. Right: Graphical representation of the GALFD $\mathfrak{N}_{20,20}$.
The LRFSs with $\theta\in[\nicefrac{\pi}{4},\nicefrac{3\pi}{4})$
are shown in red and LRFSs with $\theta\in[\nicefrac{3\pi}{4},\nicefrac{5\pi}{4})$
are shown in blue.\label{fig:GALFD} \label{fig:The-LSS}}
\end{figure}

Because of the Fourier Slice Theorem~\cite[(9.7)]{her09}, the mathematics
of CT allows us to think of the sampling obtained by a CT scanner
as equivalent to a sampling of the Fourier space. Unlike in MRI, where
direct acquisition of samples of the Fourier space is possible in
arbitrary patterns, it is not possible to obtain directly the DFT
over a CFD from CT data, due to the nature of the technique, but it
is indeed possible to obtain Fourier samples over a PFD or over an
LFD. Inversion of PFD data, however, requires certain approximations
to be made or the use of computationally more expensive operations
than just numerical Fourier transformations. On the other hand, inversion
of LFD data is not only theoretically possible~\cite{edh87}, but
is also precise and computationally effective in practice~\cite{ehr88}.
In~\cite{ahr90} LFDs appeared in the literature for the first time
in connection with an MRI application. Efficient computation of the
forward operator $\mathrm{D}_{\mathfrak{L}_{M,N}}\boldsymbol{x}$
is also discussed in~\cite{acd08a}, where the DFT over an LFD is
named the Pseudo Polar Fourier Transform.

\subsection{The Golden Angle Polar Fourier Domain}

\begin{figure}
\centering{}\animategraphics[width=0.5\columnwidth,autoplay,loop]%
{1}{gapfd_}{0}{20}\includegraphics[width=0.5\columnwidth]{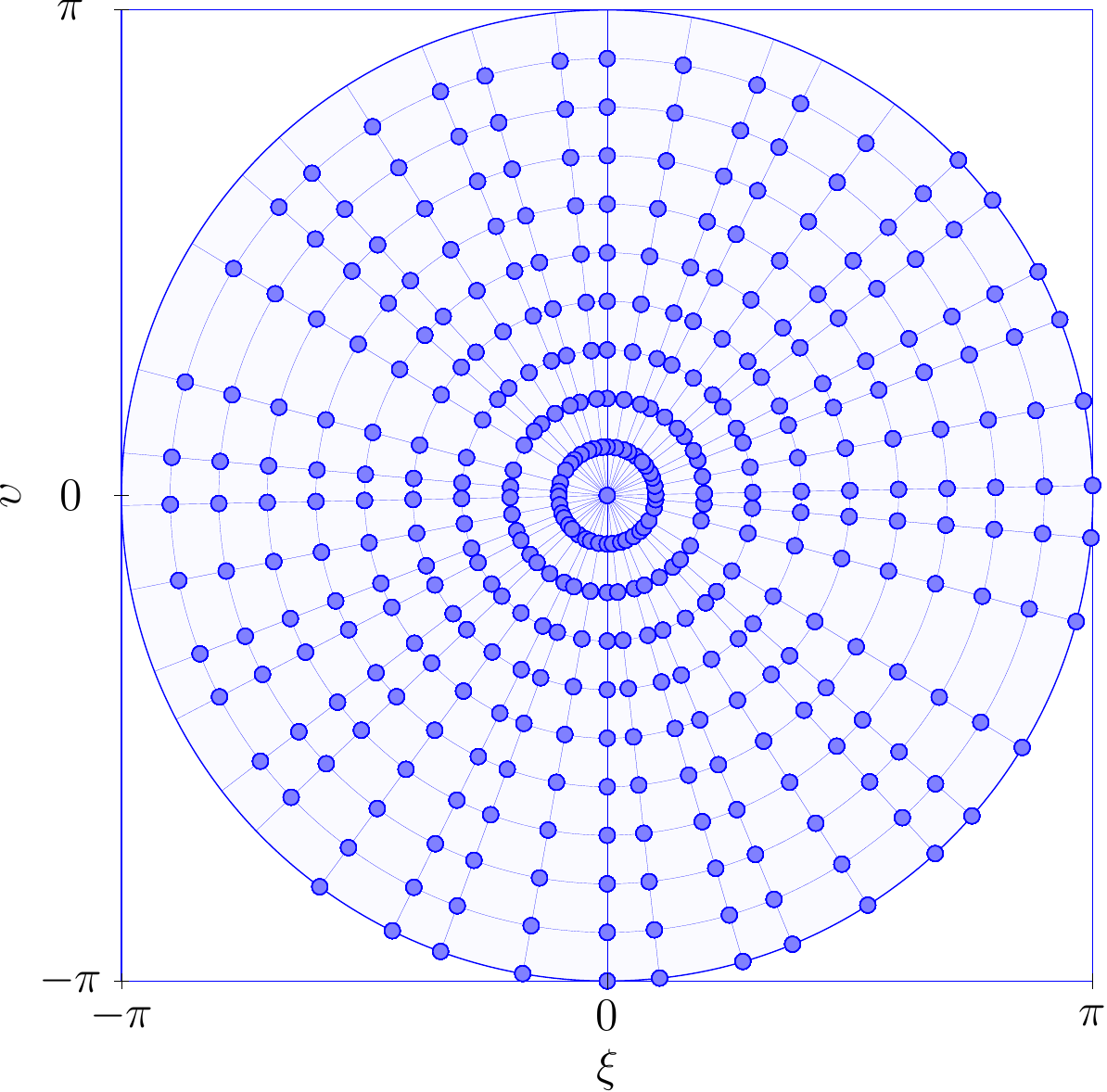}\caption{Left: Animation of the sequence of GAPFDs $\{\mathfrak{G}_{20,0},\mathfrak{G}_{20,1},\dots,\mathfrak{G}_{20,20}\}$,
showing successive inclusion of the golden-angle spaced rays (JavaScript
capable PDF viewer is required). Right: The GAPFD $\mathfrak{G}_{20,20}$.\label{fig:GAPFD-animation}}
\end{figure}

The golden angle is $\Theta:=\nicefrac{\pi}{\phi}$, where $\phi:=\nicefrac{\left(1+\sqrt{5}\right)}{2}$
is the golden ratio. We then define the GAPFD $\mathfrak{G}_{M,N,\Theta_{0}}$
by 
\begin{eqnarray*}
\fl\mathfrak{G}_{M,N,\Theta_{0}}:=\Biggl\{\Biggl(\frac{2\pi I}{M}\cos\theta_{J},\frac{2\pi I}{M}\sin\theta_{J}\Biggr):J & {}\in\{0,1,\dots,N-1\}\\
 & \text{and}\quad I\in\left\{ -\nicefrac{M}{2},-\nicefrac{M}{2}+1,\dots,\nicefrac{M}{2}-1\right\} \Biggr\},
\end{eqnarray*}
where $\theta_{J}:=\left(\Theta_{0}+J\Theta\right)$, with $\Theta_{0}\in[0,2\pi]$.
In principle, there is no reason to choose $\Theta$ to be $\nicefrac{\pi}{\phi}$.
For example, the PFD $\mathfrak{P}_{M,N}$ would be equal to a rotation
of the GAPFD $\mathfrak{G}_{M,N,\Theta_{0}}$ if $\Theta$ was replaced
by $\nicefrac{\pi}{N}$. However, a GAPFD offers the advantages that
the fine inclusion $\mathfrak{G}_{M,N,\Theta_{0}}\subset\mathfrak{G}_{M,N+1,\Theta_{0}}$
holds and that all of the GAPFDs present almost evenly distributed
rays in the Fourier space~\cite{wsk06}.

These two properties together are of practical significance. They
mean that new data can be incrementally added to an existing GAPFD
sampling, as illustrated in Figure~\ref{fig:GAPFD-animation}, in
order to improve reconstruction while maintaining the overall structure
of the dataset, a property that has been taken advantage of in several
applications, see~\cite{bth18,cld17,fgb14} and references therein
for more on this subject.

\subsection{The Golden Angle Linogram Fourier Domain}

We now present our proposed family of domains, of which each member
will be named a Golden Angle Linogram Fourier Domain (GALFD). A GALFD
is an intermediate between an LFD and a GAPFD. It shares with an LFD
the arrangement of the points of each of its rays into concentric
squares, but unlike an LFD, the rays are distributed following the
same golden angle spacing of a GAPFD. This way we retain the practical
advantages of a GAPFD in comparison to a PFD in applications, while
computation of the DFT over a GALFD can be performed more efficiently
and accurately than computation of the DFT over a GAPFD or a PFD.
How to perform the computation of the DFT over a GALFD is the subject
of the next two sections.

We first define a Linogram Ray Fourier Subdomain (LRFS) $\mathfrak{R}_{M,\sigma}(\theta)$
for a positive even integer $M$, $\sigma\in\mathbb{R}$, and $\theta\in[\nicefrac{\pi}{4},\nicefrac{5\pi}{4})$.
Then a GALFD is defined as a union of such subdomains. The definition
of an LRFS is as follows:
\begin{equation}
\fl\mathfrak{R}_{M,\sigma}(\theta)=\left\{ \begin{array}{l}
\biggl\{\biggl(\frac{2\pi I}{M}+\sigma,\left(\frac{2\pi I}{M}+\sigma\right)\tan\theta\biggr):I\in\left\{ -\nicefrac{M}{2},-\nicefrac{M}{2}+1,\dots,\nicefrac{M}{2}-1\right\} \biggr\}\\
\qquad\qquad\qquad\qquad\qquad\qquad\qquad\qquad\qquad\qquad\quad\text{if }\theta\in\left[\nicefrac{3\pi}{4},\nicefrac{5\pi}{4}\right),\\
\biggl\{\biggl(\left(\frac{2\pi I}{M}-\sigma\right)\cot\theta,\frac{2\pi I}{M}-\sigma\biggr):I\in\left\{ -\nicefrac{M}{2}+1,-\nicefrac{M}{2}+2,\dots,\nicefrac{M}{2}\right\} \biggr\}\\
\qquad\qquad\qquad\qquad\qquad\qquad\qquad\qquad\qquad\qquad\quad\ \text{if }\theta\in\left[\nicefrac{\pi}{4},\nicefrac{3\pi}{4}\right).
\end{array}\right.\label{eq:LRFS-def}
\end{equation}

Let the angle constraining operator $\Lambda:\mathbb{R}\to\left[\nicefrac{\pi}{4},\nicefrac{5\pi}{4}\right)$
be given by
\[
\Lambda(\theta):=\left(\theta-\nicefrac{\pi}{4}\right)\%\pi+\nicefrac{\pi}{4},
\]
where $a\%b$ is, for real numbers $a$ and $b$, the remainder of
the division of $a$ by $b$. The range $\left[\nicefrac{\pi}{4},\nicefrac{5\pi}{4}\right)$
is conventional; we could have chosen any half-open interval of length
$\pi$. We define the GALFD $\mathfrak{N}_{M,N,\Theta_{0},\sigma}$
by (see Figure~\ref{fig:GALFD}, right, for an example)
\[
\mathfrak{N}_{M,N,\Theta_{0},\sigma}:=\bigcup_{J=0}^{N}\mathfrak{R}_{M,\sigma}\left(\Lambda\left(\Theta_{0}+J\Theta\right)\right).
\]
Whenever possible, we simplify the notation by using $\mathfrak{N}_{M,N}:=\mathfrak{N}_{M,N,\nicefrac{\pi}{2},\nicefrac{\pi}{M}}$.

In the next sections we describe our method for efficiently computing
the DFT over a GALFD. Section~\ref{sec:Method} presents the required
mathematical development, including an upper bound on the incurred
approximation error. Section~\ref{sec:Algorithm} provides a precise
algorithmic formulation for the whole procedure and discusses floating
point operations count and memory-related issues.

\section{Mathematical Foundations\label{sec:Method}}

In the present section we discuss the mathematical foundations of
the method we propose for the computation of the DFT over an LRFS.
We give an approximation of $\mathrm{D_{\mathfrak{R}_{M,\sigma}(\theta)}}\left[\boldsymbol{x}\right]$
for any fixed $\theta\in\left[\nicefrac{\pi}{4},\nicefrac{3\pi}{4}\right)$
and $\boldsymbol{x}\in\mathbb{C}^{mn}$. We assume that $M\ge m$,
and that $M$ is even. The evenness assumption is used only for simplicity
of exposition, but $M\ge m$ is useful for computational efficiency.
The reasoning for $\theta\in\left[\nicefrac{3\pi}{4},\nicefrac{5\pi}{4}\right)$
is analogous and consequently we do not present its repetitive details.
In the algorithmic description of the next section, however, the computational
details are given in a concise but complete way.

Consider an arbitrary $\theta\in\left[\nicefrac{\pi}{4},\nicefrac{3\pi}{4}\right)$
and write $c=\cot\theta$. We proceed from the definition

\begin{eqnarray*}
\fl\mathcal{D}[\boldsymbol{x}]\left(\left(\frac{2\pi I}{M}-\sigma\right)c,\frac{2\pi I}{M}-\sigma\right) & {}=\sum_{i=0}^{m-1}\sum_{j=0}^{n-1}x_{i,j}e^{-\imath\left(j\left(\frac{2\pi I}{M}-\sigma\right)c+i\left(\frac{2\pi I}{M}-\sigma\right)\right)}\\
 & {}=\sum_{j=0}^{n-1}\sum_{i=0}^{m-1}x_{i,j}e^{-\imath\left(j\left(\frac{2\pi I}{M}-\sigma\right)c+i\left(\frac{2\pi I}{M}-\sigma\right)\right)}\\
 & {}=\sum_{j=0}^{n-1}e^{-\imath j\left(\frac{2\pi I}{M}-\sigma\right)c}\sum_{i=0}^{m-1}x_{i,j}e^{-\imath i\left(\frac{2\pi I}{M}-\sigma\right)}\\
 & {}=\sum_{j=0}^{n-1}X_{I,j}e^{-\imath j\left(\frac{2\pi I}{M}-\sigma\right)c},
\end{eqnarray*}
where (note that this is a one-dimensional discrete Fourier transform
of $x_{i,j}$ with respect to $i$, refer to~\ref{sec:FFT} for details)
\begin{equation}
X_{I,j}:=\sum_{i=0}^{m-1}x_{i,j}e^{\imath i\sigma}e^{-\imath i\frac{2\pi I}{M}}.\label{eq:XIj}
\end{equation}
Now we rewrite these computations as
\begin{equation}
\mathcal{D}[\boldsymbol{x}]\left(\left(\frac{2\pi I}{M}-\sigma\right)c,\frac{2\pi I}{M}-\sigma\right)=\sum_{j=0}^{n-1}X_{I,j}e^{-\imath j\frac{2\pi\eta}{N_{L}}\alpha_{I}},\label{eq:DTFT-eta}
\end{equation}
where $N_{L}\ge2n$ is an integer divisible by $4$, $\eta=\nicefrac{cN_{L}}{4}$
and $\alpha_{I}=\nicefrac{4I}{M}-\nicefrac{2\sigma}{\pi}$. If $\eta$
were an integer with $\eta\in\left\{ -\nicefrac{N_{L}}{4},-\nicefrac{N_{L}}{4}+1,\dots,\nicefrac{N_{L}}{4}-1\right\} $,
then the right hand side of (\ref{eq:DTFT-eta}) would be a part of
the output of a Chirp-Z Transform (CZT)~\cite{blu70} of length $\nicefrac{N_{L}}{2}$,
which can be efficiently computed with the use of one FFT and one
Inverse FFT (IFFT)\footnote{We provide the definition of a mild generalization of the CZT and
describe an algorithm for its computation in detail in \ref{sub:CZT}.}. This is what makes the computation of the DFT over an LFD efficient.
However, in the LRFS case that we are considering, where $c=\cot\theta$
is not constrained to specific values in the interval {[}-1,1{]},
$\eta$ is not necessarily an integer, and we need an appropriate
theoretical tool to help us to deal with this situation.

As a preliminary step, we provide the definition of the Continuous
Fourier Transform (CFT). Let $f:\mathbb{R}^{n}\to\mathbb{C}$ be an
absolutely integrable function, its CFT $\hat{f}:\mathbb{R}^{n}\to\mathbb{C}$
is 
\[
\hat{f}(\boldsymbol{\omega}):=\mathcal{F}[f](\boldsymbol{\omega}):=\int_{\mathbb{R}^{n}}f(\boldsymbol{x})e^{-\imath\langle\boldsymbol{x},\boldsymbol{\omega}\rangle}\mathrm{d}\boldsymbol{x},
\]
where $\langle\boldsymbol{x},\boldsymbol{\omega}\rangle$ is the inner
product between $\boldsymbol{x}\in\mathbb{R}^{n}$ and $\boldsymbol{\omega}\in\mathbb{R}^{n}$.
The following is a very useful result due to Fourmont~\cite{fou03}.
Since the differences in our statement are due only to the notation
and to the normalization of the CFT, we refer the reader to~\cite[Proposition~1]{fou03}
for a proof.
\begin{theorem}
\label{thm:Fourmont}Let $\tau\in\mathbb{R}$ and $\delta\in\mathbb{R}$
be such that $0<\nicefrac{\pi}{\delta}<\tau<2\pi-\nicefrac{\pi}{\delta}$.
Let also $W:\mathbb{R}\to\mathbb{R}$ be continuous and piecewise
continuously differentiable in $[-\tau,\tau]$, vanishing outside
$[-\tau,\tau]$ and nonzero in $\left[-\nicefrac{\pi}{\delta},\nicefrac{\pi}{\delta}\right]$,
and let $\hat{W}:\mathbb{R}\to\mathbb{R}$ be its Fourier transform.
Then, for $\zeta\in\mathbb{R}$ and $|t|\leq\nicefrac{\pi}{\delta}$,
we have
\[
e^{-\imath\zeta t}=\frac{1}{2\pi W(t)}\sum_{J\in\mathbb{Z}}\hat{W}(\zeta-J)e^{-\imath Jt}.
\]

\end{theorem}
\par{}As an immediate consequence of the above equality, we have,
for $|t-\varpi|\leq\pi/\delta$,
\begin{eqnarray}
e^{-\imath\zeta t} & {}=e^{-\imath\zeta\varpi}e^{-\imath\zeta(t-\varpi)}\nonumber \\
 & {}=\frac{e^{-\imath\zeta\varpi}}{2\pi W(t-\varpi)}\sum_{J\in\mathbb{Z}}\hat{W}(\zeta-J)e^{-\imath J(t-\varpi)}\nonumber \\
 & {}=\frac{1}{2\pi W(t-\varpi)}\sum_{J\in\mathbb{Z}}\hat{W}(\zeta-J)e^{-\imath(\zeta-J)\varpi}e^{-\imath Jt}.\label{eq:Fourmont-series}
\end{eqnarray}

Now, let $t_{j}^{I}:=\left(\nicefrac{(2\pi j)}{N_{L}}\right)\alpha_{I}$
(where $\alpha_{I}$ is as in (\ref{eq:DTFT-eta})) and 
\begin{equation}
\varpi_{I}:=\frac{\pi(n-1)}{N_{L}}\alpha_{I}.\label{eq:varpi_I}
\end{equation}
At this point, we must introduce the constraint $|\sigma|<\nicefrac{\pi}{(n-1)}$.
If this condition is true, taking into consideration that $N_{L}\ge2n,$
that $j\in\{0,1,\dots,n-1\}$, and that $|I|\le\nicefrac{M}{2}$,
we have
\begin{eqnarray}
|t_{j}^{I}-\varpi_{I}| & {}=\left|\frac{\pi\alpha_{I}}{N_{L}}\left(2j-(n-1)\right)\right|=\frac{\pi}{N_{L}}\left|\frac{4I}{M}-\frac{2\sigma}{\pi}\right|\left|2j-(n-1)\right|\nonumber \\
 & {}\leq\frac{\pi}{N_{L}}\left|\frac{4I}{M}-\frac{2\sigma}{\pi}\right|(n-1)\le\frac{\pi}{2}\left|\frac{4I}{M}-\frac{2\sigma}{\pi}\right|\frac{n-1}{n}\nonumber \\
 & {}\le\frac{\pi}{2}\left(\left|\frac{4I}{M}\right|+\frac{2|\sigma|}{\pi}\right)\frac{n-1}{n}\le\frac{\pi}{2}\left(2+\frac{2|\sigma|}{\pi}\right)\frac{n-1}{n}\nonumber \\
 & {}=\pi\left(1+\frac{|\sigma|}{\pi}\right)\frac{n-1}{n}<\pi\left(1+\frac{1}{n-1}\right)\frac{n-1}{n}=\pi.\label{eq:t-range}
\end{eqnarray}

Now, define, for an $\varepsilon\in(0,1)$ such that $\varepsilon\approx1$
(we use the specific value $\varepsilon=1-10^{-4}$),
\begin{equation}
\tau_{I}:=\pi+\varepsilon\left(\pi-|\varpi_{I}|\right).\label{eq:tau_I}
\end{equation}
Note that a computation similar to~\eqref{eq:t-range} gives $|\varpi_{I}|<\pi$;
therefore, because of the definition of $\tau_{I}$, we have $\pi<\tau_{I}<2\pi$.
Thus, if some $W_{I}:\mathbb{R}\to\mathbb{R}$ is smooth for $t\in\left[-\tau_{I},\tau_{I}\right]$,
$W_{I}(t)\neq0$ for $t\in\left(-\tau_{I},\tau_{I}\right)$, and $W_{I}(t)=0$
for $t\notin\left(-\tau_{I},\tau_{I}\right)$, then $W_{I}$ satisfies
all conditions for $W$ in the statement of Theorem~\ref{thm:Fourmont}
with any $\delta>\max\left\{ \nicefrac{\pi}{\tau_{I}},\nicefrac{\pi}{\left(2\pi-\tau_{I}\right)}\right\} $.
Hence, in view of~\eqref{eq:t-range},~\eqref{eq:Fourmont-series}
can be used with $W=W_{I}$, $t=t_{j}^{I}$, $\zeta=\eta$, and $\varpi=\varpi_{I}$
inside of~\eqref{eq:DTFT-eta} in order to yield
\begin{eqnarray}
\fl\mathcal{D}[\boldsymbol{x}]\Biggl(\left(\frac{2\pi I}{M}-\sigma\right)c,\frac{2\pi I}{M} & {}-\sigma\Biggr)\nonumber \\
 & {}=\sum_{j=0}^{n-1}X_{I,j}e^{-\imath\eta t_{j}^{I}}\nonumber \\
 & {}=\sum_{j=0}^{n-1}X_{I,j}\frac{1}{2\pi W_{I}\left(t_{j}^{I}-\varpi_{I}\right)}\sum_{J\in\mathbb{Z}}\hat{W}_{I}(\eta-J)e^{-\imath(\eta-J)\varpi_{I}}e^{-\imath Jt_{j}^{I}}\nonumber \\
 & {}=\frac{1}{2\pi}\sum_{J\in\mathbb{Z}}\hat{W}_{I}(\eta-J)e^{-\imath(\eta-J)\varpi_{I}}\sum_{j=0}^{n-1}X_{I,j}\frac{1}{W_{I}\left(t_{j}^{I}-\varpi_{I}\right)}e^{-\imath Jt_{j}^{I}}\nonumber \\
 & {}=\frac{1}{2\pi}\sum_{J\in\mathbb{Z}}\frac{\hat{W}_{I}(\eta-J)}{e^{\imath(\eta-J)\varpi_{I}}}\sum_{j=0}^{n-1}X_{I,j}\frac{1}{W_{I}\left(t_{j}^{I}-\varpi_{I}\right)}e^{-\imath j\frac{2\pi J}{N_{L}}\alpha_{I}}.\label{eq:nontruncated-sum}
\end{eqnarray}
This means that the values of the DTFT at the desired points can be
approximated by a sum
\begin{eqnarray}
\fl\mathcal{D}[\boldsymbol{x}]\Biggl(\left(\frac{2\pi I}{M}-\sigma\right)c,\frac{2\pi I}{M} & {}-\sigma\Biggr)\nonumber \\
 & {}\approx\frac{1}{2\pi}\sum_{|J-\eta|\leq S}\frac{\hat{W}_{I}(\eta-J)}{e^{\imath(\eta-J)\varpi_{I}}}\sum_{j=0}^{n-1}X_{I,j} & \frac{1}{W_{I}\left(t_{j}^{I}-\varpi_{I}\right)}e^{-\imath j\frac{2\pi J}{N_{L}}\alpha_{I}}\nonumber \\
 &  & \qquad\qquad\ {}=:D_{I,N_{L},S}(\theta),\label{eq:truncated-sum}
\end{eqnarray}
where the positive integer $S$, which dictates the number of terms
in the truncated sum, is a parameter determining the accuracy of the
approximation. Another such parameter is the length of the support
of $W_{I}$. We return to this topic when we derive an upper bound
for the approximation error.

Now we tackle the issue of determining the ``window functions''
$W_{I}:\mathbb{R}\to\mathbb{R}$ that satisfy the conditions of Theorem~\ref{thm:Fourmont}.
Of course, there is some freedom in the selection of these window
functions and many have been proposed and used in the literature.
Among the most successful and popular is the family of Kaiser-Bessel~\cite{kas80}
Window Functions (KBWFs). Members of this family of functions have
finite support as required for the theorem to hold, and they also
have the desirable property that the infinite series on the right-hand
side of~\eqref{eq:Fourmont-series} converges very quickly, thus
making the approximation by a finite summation precise even when only
a few terms are used. The definition of a KBWF is as follows. For
parameters $\beta\in\mathbb{R}$ and $\tau\in\mathbb{R}$,
\[
K_{\beta,\tau}(t):=\left\{ \begin{array}{ll}
{\displaystyle \frac{\mathcal{I}_{0}\left(\beta\sqrt{1-(t/\tau)^{2}}\right)}{\mathcal{I}_{0}(\beta)},} & \quad\text{if}\quad|t|\leq\tau,\\
0, & \quad\text{otherwise,}
\end{array}\right.
\]
where $\mathcal{I}_{0}:\mathbb{R}\to\mathbb{R}$ is the (real) modified
zero-order cylindrical Bessel function. Note that
\[
\hat{K}_{\beta,\tau}(\omega)=\frac{2\tau}{\mathcal{I}_{0}(\beta)}\frac{\sinh\left(\beta\sqrt{1-(\omega\tau/\beta)^{2}}\right)}{\beta\sqrt{1-(\omega\tau/\beta)^{2}}}.
\]
In our methods we use specifically the following window functions:
For $\beta_{I}:=S\tau_{I}$, 
\begin{equation}
W_{I}:=K_{\beta_{I},\tau_{I}}.\label{eq:window-definition}
\end{equation}
Figure~\ref{fig:Plots-of-Kaiser-Bessel} shows examples of the Kaiser-Bessel
window functions described above for some combination of parameters.

\begin{figure}
\includegraphics[width=0.5\columnwidth]{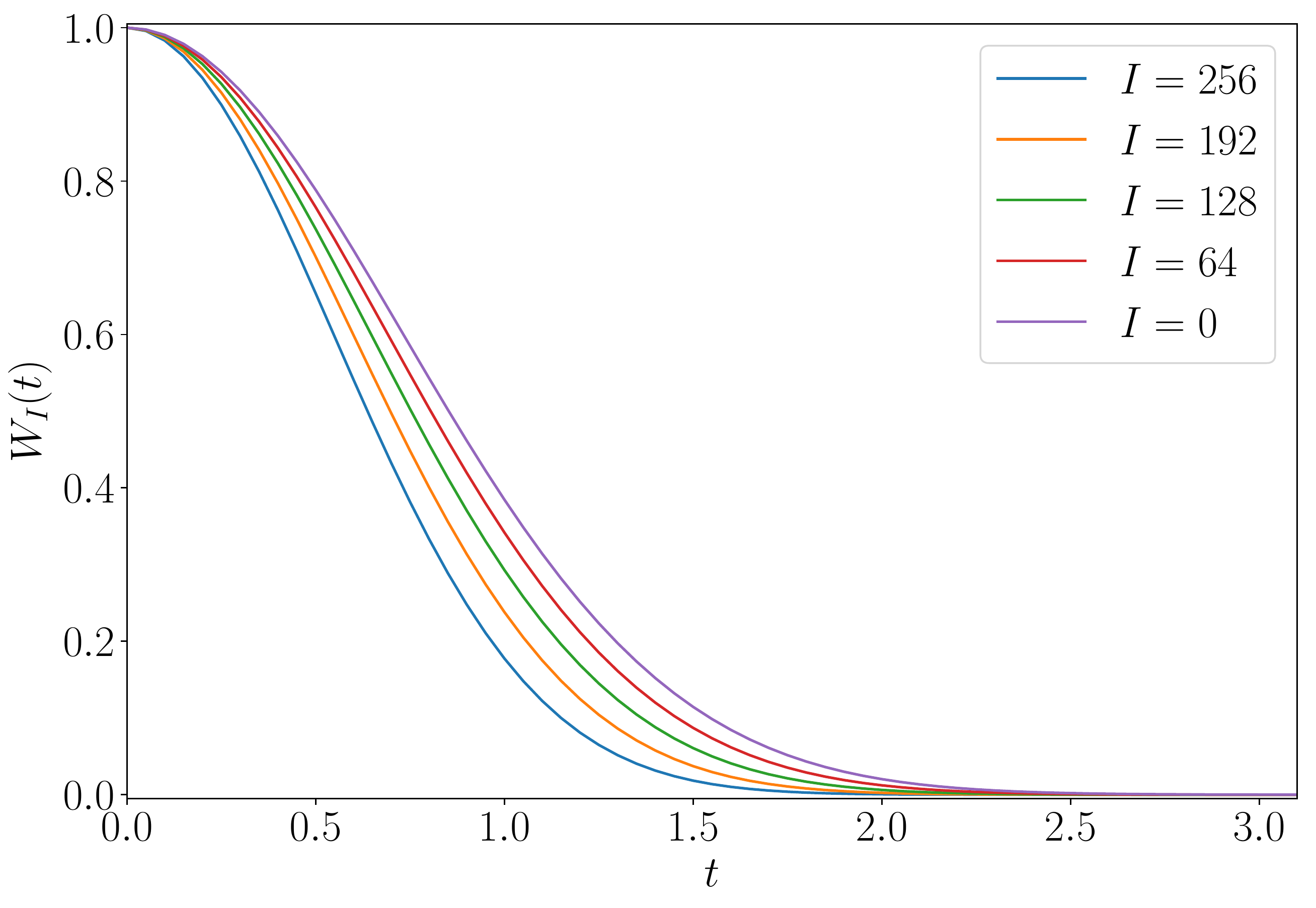}\includegraphics[width=0.5\columnwidth]{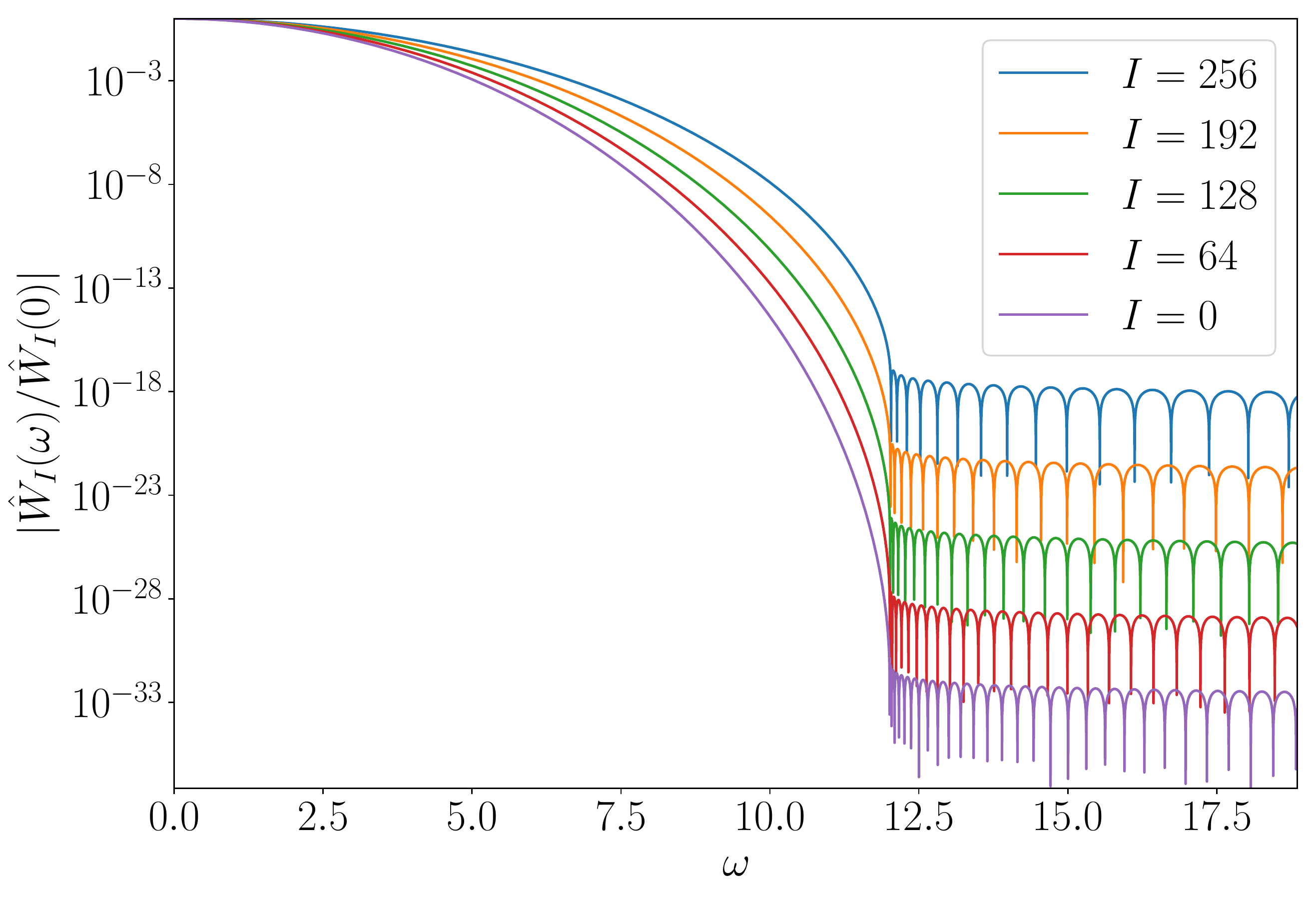}\caption{Plots of Kaiser-Bessel window functions (left) and respective Fourier
transforms (right) for several values of $I$ with parameters $m=n=M=512$,
$N_{L}=2.25m$, $\sigma=\nicefrac{\pi}{M}$, and $\epsilon=1-10^{-4}$.\label{fig:Plots-of-Kaiser-Bessel}}
\end{figure}

We now present our upper bound for the approximation error due to
the truncation~\eqref{eq:truncated-sum}.
\begin{theorem}
Let $\theta\in\left[\nicefrac{\pi}{4},\nicefrac{3\pi}{4}\right),$
$|\sigma|<\nicefrac{\pi}{(n-1)}$, $\boldsymbol{x}\in\mathbb{C}^{mn}$,
$M\ge m$, $N_{L}\ge2n$, $1<S\le15$, $|I|\le\nicefrac{M}{2}$, $\varpi_{I}$
as in~\eqref{eq:varpi_I}, $\tau_{I}$ as in~\eqref{eq:tau_I},
and $D_{I,N_{L},S}(\theta)$ as in~\eqref{eq:truncated-sum} with~\eqref{eq:window-definition},
then
\begin{equation}
\fl\left|\mathcal{D}[\boldsymbol{x}]\left(\left(\frac{2\pi I}{M}-\sigma\right)\cot\theta,\frac{2\pi I}{M}-\sigma\right)-D_{I,N_{L},S}(\theta)\right|\le\frac{29.5\|\boldsymbol{x}\|_{1}}{\pi\mathcal{I}_{0}\left(S\sqrt{\tau_{I}^{2}-\varpi_{I}^{2}}\right)},\label{eq:error-bound}
\end{equation}
where $\|\boldsymbol{x}\|_{1}:=\sum_{i=0}^{m-1}\sum_{j=0}^{n-1}\left|x_{i,j}\right|$.\end{theorem}
\begin{proof}
Let us first consider the difference below (notice that, although
omitted from the notation, $\eta$ is a function of $N_{L}$ and $\theta$,
and $W_{I}$ depends on $N_{L}$ through the parameter $\tau_{I}$):
\begin{equation}
\fl\varepsilon_{j}^{I,N_{L},S}(\theta):=e^{-\imath\eta t_{j}^{I}}-\frac{1}{2\pi W_{I}\left(t_{j}^{I}-\varpi_{I}\right)}\sum_{|J-\eta|\le S}\hat{W}_{I}(\eta-J)e^{-\imath(\eta-J)\varpi_{I}}e^{-\imath Jt_{j}^{I}}.\label{eq:eps_def}
\end{equation}
Notice that $\varepsilon_{j}^{I,N_{L},S}(\theta)$ is well defined
since \eqref{eq:t-range} and \eqref{eq:tau_I} imply that $\left|t_{j}^{I}-\varpi_{I}\right|<\tau_{I}$,
which in turn, following definition~\eqref{eq:window-definition},
implies that $W_{I}\left(t_{j}^{I}-\varpi_{I}\right)>0$. Because
of~\eqref{eq:Fourmont-series} we have
\begin{eqnarray*}
\left|\varepsilon_{j}^{I,N_{L},S}(\theta)\right| & {}=\left|\frac{1}{2\pi W_{I}\left(t_{j}^{I}-\varpi_{I}\right)}\sum_{|J-\eta|>S}\hat{W}_{I}(\eta-J)e^{-\imath(\eta-J)\varpi_{I}}e^{-\imath Jt_{j}^{I}}\right|\\
 & {}=\frac{1}{2\pi W_{I}\left(t_{j}^{I}-\varpi_{I}\right)}\left|\sum_{|J-\eta|>S}\hat{W}_{I}(\eta-J)e^{\imath(\eta-J)(t_{j}^{I}-\varpi_{I})}\right|,
\end{eqnarray*}
since $\left|e^{-\imath\eta t_{j}^{I}}\right|=1$. By denoting $\ell=\eta-J$
and $t=t_{j}^{I}-\varpi_{I}$ we get
\begin{equation}
\left|\varepsilon_{j}^{I,N_{L},S}(\theta)\right|=\frac{1}{2\pi W_{I}(t)}\left|\sum_{|\ell|>S}\hat{W}_{I}(\ell)e^{\imath\ell t}\right|.\label{eq:simplified-difference}
\end{equation}

Now consider the specific window functions we use. Because $W_{I}=K_{\beta_{I},\tau_{I}}$
where $\beta_{I}=S\tau_{I}$,
\[
\hat{W}_{I}(\omega)=\frac{2}{\mathcal{I}_{0}(\beta_{I})}\frac{\sinh\left(\tau_{I}\sqrt{S^{2}-\omega^{2}}\right)}{\sqrt{S^{2}-\omega^{2}}}.
\]
For the specific case where $\ell>S$, the square roots are imaginary
and we can write this as
\[
\hat{W}_{I}(\ell)=\frac{2}{\mathcal{I}_{0}(\beta_{I})}\frac{\sin\left(\tau_{I}\sqrt{\ell^{2}-S^{2}}\right)}{\sqrt{\ell^{2}-S^{2}}}.
\]
 Thus, using this equality in~\eqref{eq:simplified-difference} and
considering that $W_{I}=K_{\beta_{I},\tau_{I}}$ we get, because $|t|<\tau_{I}$,
\[
\left|\varepsilon_{j}^{I,N_{L},S}(\theta)\right|=\frac{1}{\pi\mathcal{I}_{0}\left(S\sqrt{\tau_{I}^{2}-t^{2}}\right)}\left|\sum_{|\ell|>S}\frac{\sin\left(\tau_{I}\sqrt{\ell^{2}-S^{2}}\right)}{\sqrt{\ell^{2}-S^{2}}}\right|.
\]
Now we can apply~\cite[Corollar 2.5.12]{fou99}, which holds for
$1<S\leq15$, in order to obtain
\[
\left|\varepsilon_{j}^{I,N_{L},S}(\theta)\right|\le\frac{29.5}{\pi\mathcal{I}_{0}\left(S\sqrt{\tau_{I}^{2}-t^{2}}\right)}.
\]
Now, notice that because the Bessel function $\mathcal{I}_{0}$ is
monotonically increasing, the denominator is minimized when $t^{2}=\left(t_{j}^{I}-\varpi_{I}\right)^{2}$
is as close as possible to $\tau_{I}^{2}$. Equivalently, because
$\left(t_{j}^{I}-\varpi_{I}\right)^{2}<\pi^{2}<\tau_{I}^{2}$ (see
\eqref{eq:t-range} and the definition of $\tau_{I}$ in (\ref{eq:tau_I})),
this occurs when $\left(t_{j}^{I}-\varpi_{I}\right)^{2}$ is maximized
over $j\in\{0,1,\dots,n-1\}$, which occurs when $j=0$ or when $j=(n-1)$
because
\[
t_{j}^{I}-\varpi_{I}=\frac{\pi\alpha_{I}}{N_{L}}\left(2j-(n-1)\right).
\]
From the definition of $\varpi_{I}$ in (\ref{eq:varpi_I}) it follows
that 
\[
\left(t_{0}^{I}-\varpi_{I}\right)^{2}=\left(t_{(n-1)}^{I}-\varpi_{I}\right)^{2}=\varpi_{I}^{2}.
\]
Therefore we obtain the following upper bound for $\varepsilon_{j}^{I,N_{L},S}(\theta)$
that is uniform in $\theta$ and $j$:
\[
\left|\varepsilon_{j}^{I,N_{L},S}(\theta)\right|\le\frac{29.5}{\pi\mathcal{I}_{0}\left(S\sqrt{\tau_{I}^{2}-\varpi_{I}^{2}}\right)}=:\varepsilon_{I,N_{L},S}.
\]

Defining $\delta_{j}^{I,N_{L},S}(\theta)$ as below and using~\eqref{eq:nontruncated-sum},
\eqref{eq:truncated-sum}, and~\eqref{eq:eps_def} we complete the
proof:
\[
\delta_{j}^{I,N_{L},S}(\theta):=\left|\mathcal{D}[\boldsymbol{x}]\biggl(\biggl(\frac{2\pi I}{M}-\sigma\biggr)\cot\theta,\frac{2\pi I}{M}-\sigma\biggr)-D_{I,N_{L},S}(\theta)\right|.
\]
\begin{eqnarray*}
\delta_{j}^{I,N_{L},S}(\theta) & {}=\left|\sum_{j=0}^{n-1}X_{I,j}\varepsilon_{j}^{I,N_{L},S}(\theta)\right|\\
 & {}\leq\left|\sum_{j=0}^{n-1}X_{I,j}\right|\varepsilon_{I,N_{L},S}\le\sum_{j=0}^{n-1}\left|X_{I,j}\right|\varepsilon_{I,N_{L},S}\\
 & {}=\sum_{j=0}^{n-1}\left|\sum_{i=0}^{m-1}x_{i,j}e^{\imath i\sigma}e^{-\imath i\frac{2\pi I}{M}}\right|\varepsilon_{I,N_{L},S}\\
 & {}\le\varepsilon_{I,N_{L},S}\sum_{j=0}^{n-1}\sum_{i=0}^{m-1}\left|x_{i,j}\right|.
\end{eqnarray*}
 
\end{proof}
\par{}In the next section we give details of how to turn the above
development into practical algorithms. Here we discuss some aspects
of the error bound that we have just obtained. Consider fixed $N_{L}$
and $S$. We notice that $\tau_{I}^{2}-\varpi_{I}^{2}$ is larger
for the lower frequency elements of the GALFD (i.e., $I$ close to
$0$). This observation is important because the error bound in~\eqref{eq:error-bound}
is smaller when $\tau_{I}^{2}-\varpi_{I}^{2}$ is larger. This means
that, for a given computational effort (which, for fixed $m$, $n$,
and $M$, depends only on $N_{L}$ and $S$), we get better accuracy
in the region of the spectrum where data are usually more reliable
in practical applications. This is also interesting in the scenario
where one wishes to vary $N_{L}$ and $S$ for each $I$, which could
be easily considered in our theoretical framework, in order to obtain
an overall desired accuracy with the smallest possible computational
cost. For example, Figure~\ref{fig:error-bound} shows that the Fourier
sample length $N_{L}$ becomes less important as $I$ gets close to
zero, and could therefore be reduced in these regions, because there
is a ``natural oversampling'' due to the fact that the samples used
in the truncated sum get very close to each other in that region.
Our approximation makes use of this fact by lengthening the size of
the $W_{I}$ window support, thereby improving its decay in the Fourier
domain and consequently reducing the precision loss caused by truncation
of the infinite summation.

\begin{figure}
\centering{}\includegraphics[width=0.75\columnwidth]{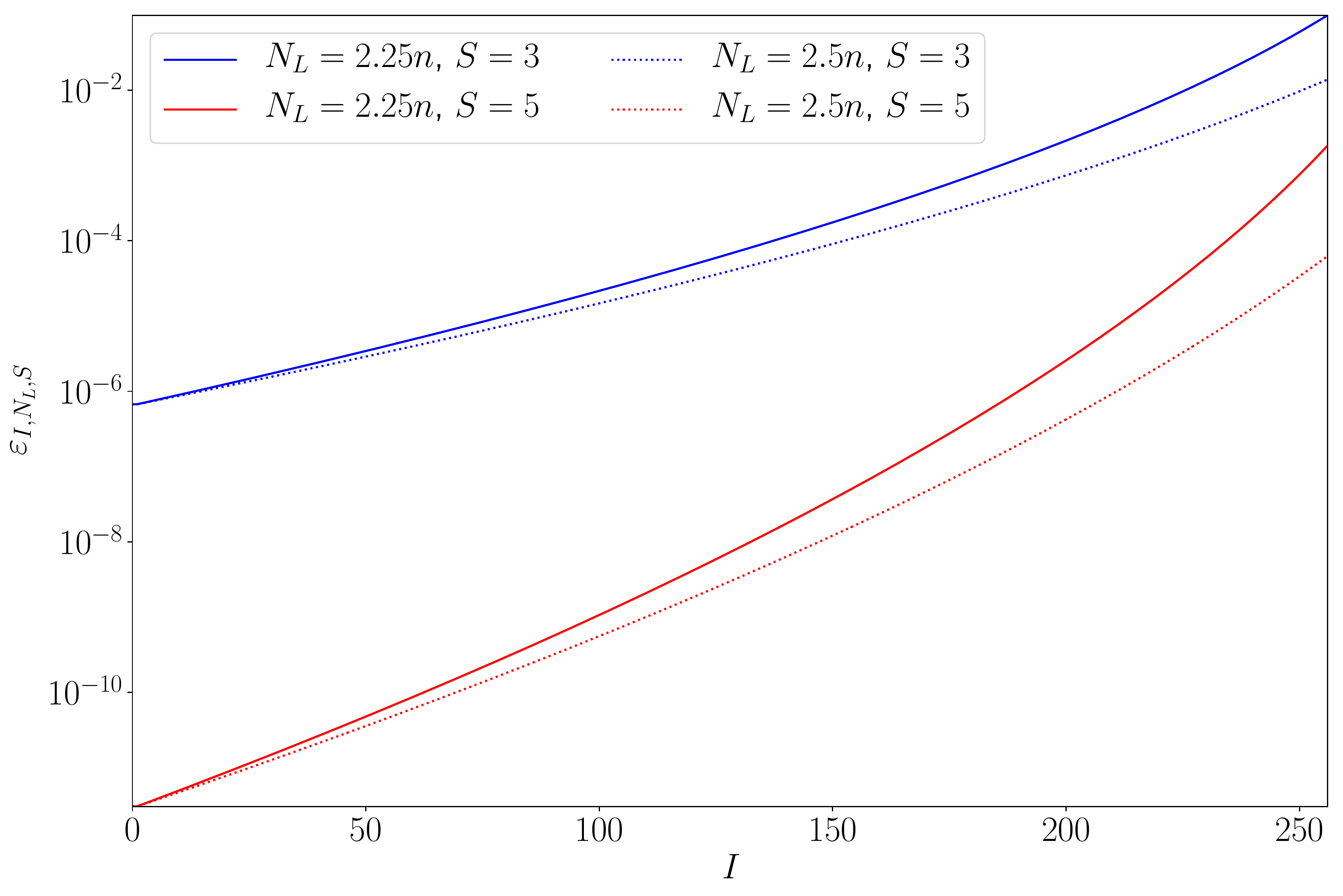}\caption{The error bound $\varepsilon_{I,N_{L},S}$ as a function of $I$ for
indicated $N_{L}$ and $S$ ($M=m=n=512$, $\sigma=\pi/M$).\label{fig:error-bound}}
\end{figure}

\section{Computational Algorithm Formulation\label{sec:Algorithm}}

At this point, the core of our strategy has been outlined. The first
step consists of the computation of the $X_{I,j}$ for $I\in\{-\nicefrac{M}{2}+1,-\nicefrac{M}{2}+2,\dots,\nicefrac{M}{2}\}$
and $j\in\{0,1,\dots,n-1\}$ using $n$ FFTs of length $M$, see \eqref{eq:XIj},
and then multiplying the results by the weights $\nicefrac{1}{W_{I}\left(t_{j}^{I}-\varpi_{I}\right)}$,
see \eqref{eq:nontruncated-sum}. After this multiplication, the values
\begin{equation}
Z_{I,J}:=\sum_{j=0}^{n-1}X_{I,j}\frac{1}{W_{I}\left(t_{j}^{I}-\varpi_{I}\right)}e^{-\imath j\frac{2\pi J}{N_{L}}\alpha_{I}}\label{eq:ZIJ}
\end{equation}
can be computed efficiently using $M$ CZTs, as we discuss below in
\ref{sub:CZT}. An important detail that should not be overlooked
is that, because $\eta$ can be very close to $\pm\nicefrac{N_{L}}{4}$,
we may need $S+1$ extra points at each end of the output of the CZT
for the right-hand side of~\eqref{eq:truncated-sum} to be computable,
and, thus, each CZT will have length $\nicefrac{N_{L}}{2}+2(S+1)$.
Finally, the values of the estimates of the DFT at the desired points
are computed by a linear combination of the previously-obtained values
of $Z_{I,J}$, as in \eqref{eq:truncated-sum}. In the appendices
we describe for completeness two well known key computational elements
that are used in our method. We leave the full description of our
algorithm to Subsection~\ref{sub:algo} and we discuss its computational
complexity in Subsection~\ref{sub:flops}.

\subsection{GALE: Approximating the DFT over a GALFD\label{sub:algo}}

Now we present GALE, the method that we propose for the computation
of the DFT over a GALFD (an open source reference implementation can
be downloaded from \url{https://bitbucket.org/eshneto/gale/downloads/}).
We give an algorithm for the region $\theta\in\left[\nicefrac{\pi}{4},\nicefrac{3\pi}{4}\right)$
and explain later how to make use of the same computations in order
to obtain the case $\theta\in\left[\nicefrac{3\pi}{4},\nicefrac{5\pi}{4}\right)$.
Note that GALE is adaptable to other domains as well.

Before starting the presentation, we discuss some typographical conventions
that are used in the description of the method. We denote non-scalar
entities such as vectors, two-dimensional, or three-dimensional arrays
of numbers by boldface letters, possibly with a subscript or with
a superscript. If $\boldsymbol{x}\in\mathbb{C}^{mn}$ is a two-dimensional
array with $m$ lines and $n$ columns, then $\boldsymbol{x}^{j}\in\mathbb{C}^{m}$
with $j\in\{0,1,\dots,n-1\}$ is the $j$-th column of $\boldsymbol{x}$,
that is, $\boldsymbol{x}^{j}=\left(x_{0,j},x_{1,j},\dots,x_{m-1,j}\right)^{T}$.
Likewise, $\boldsymbol{x}_{i}\in\mathbb{C}^{n}$ with $i\in\{0,1,\dots,m-1\}$
is the $i$-th row of $\boldsymbol{x}$ expressed as a column vector,
that is, $\boldsymbol{x}_{i}=\left(x_{i,0},x_{i,1},\dots,x_{i,n-1}\right)^{T}$.
Whenever there is a set of scalars indexed with subscripts and/or
superscripts, a corresponding boldface letter represents the collection
of those scalars. For example, $\boldsymbol{x}\in\mathbb{C}^{mn}$
is an entity that collects all scalars $x_{i,j}$. In another example,
the algorithm may set values of $v_{i}^{I,K}$ for certain ranges
of $I$, $K$, and $i$. In this case, $\boldsymbol{v}$ represents
the set of all $v_{i}^{I,K}$ that were defined during the algorithm's
execution.

\begin{algorithm}[t]
\caption{$\text{GALE-I}\left(\boldsymbol{\theta},m,n,M,N,N_{L},S,R_{1},R_{2},\sigma,\varepsilon\right)$}

\label{algo:GALDFT-INIT}

\normalsize{

\begin{algorithmic}[1]

\STATE{\textbf{for} $i\in\{0,1,\dots,m-1\}$ \textbf{do}\label{step:p-init-start}}

\STATE{$\quad\quad$$p_{i}\leftarrow e^{\imath i\left(\frac{2\pi R_{1}}{M}+\sigma\right)}$\label{step:p-init-finish}}

\STATE{\textbf{for} $I\in\{0,1,\dots,M-1\}$ \textbf{do}}

\STATE{$\quad\quad$$\alpha_{I}\leftarrow\frac{4\left(I-R_{1}\right)}{M}-\frac{2\sigma}{\pi}$\label{step:GALFD-alpha-compute}}

\STATE{$\quad\quad$$(\hat{\boldsymbol{q}}^{I},\boldsymbol{r}^{I},\boldsymbol{s}^{I})\leftarrow\text{CZT-I}\left(n,\alpha_{I},N_{L},\nicefrac{N_{L}}{2}+2(S+1),R_{2}\right)$\label{step:GALFD-CZT-init}}

\STATE{$\quad\quad$$\varpi_{I}\leftarrow\frac{\pi(n-1)}{N_{L}}\alpha_{I}$}

\STATE{$\quad\quad$$\tau_{I}\leftarrow\pi+\varepsilon\left(\pi-|\varpi_{I}|\right)$}

\STATE{$\quad\quad$\textbf{for} $j\in\{0,1,\dots,n-1\}$ \textbf{do}}

\STATE{$\quad\quad\quad\quad$$t_{j}^{I}\leftarrow\frac{2\pi j}{N_{L}}\alpha_{I}$}

\STATE{$\quad\quad\quad\quad$$r_{j}^{I}\leftarrow\frac{r_{j}^{I}}{K_{S\tau_{I},\tau_{I}}\left(t_{j}^{I}-\varpi_{I}\right)}$\label{step:GALFD-mod-r}}

\STATE{$\quad\quad$\textbf{for} $K\in\{0,1,\dots,N-1\}$ \textbf{do}}

\STATE{$\quad\quad\quad\quad$$\eta_{K}\leftarrow\frac{N_{L}\cot\theta_{K}}{4}$}

\STATE{$\quad\quad\quad\quad$$i\leftarrow0$}

\STATE{$\quad\quad\quad\quad$\textbf{for} $J\in\left\{ \ell:\left|\ell-\eta_{K}\right|\le S\right\} $
\textbf{do}}

\STATE{$\quad\quad\quad\quad\quad\quad$$v_{i}^{I,K}\leftarrow J+R_{2}$\label{step:GALFD-v-init}}

\STATE{$\quad\quad\quad\quad\quad\quad$$w_{i}^{I,K}\leftarrow\frac{\hat{K}_{S\tau_{I},\tau_{I}}\left(\eta_{K}-J\right)}{2\pi e^{\imath\left(\eta_{K}-J\right)\varpi_{I}}}$\label{step:GALFD-w-init}}

\STATE{$\quad\quad\quad\quad\quad\quad$$i\leftarrow i+1$}

\STATE{$\quad\quad\quad\quad$$u_{I,K}\leftarrow i$\label{step:GALFD-u-init}}

\STATE{\textbf{return} $(\boldsymbol{p},\hat{\boldsymbol{q}},\boldsymbol{r},\boldsymbol{s},\boldsymbol{u},\boldsymbol{v},\boldsymbol{w})$}

\end{algorithmic}}
\end{algorithm}

\begin{algorithm}[t]
\caption{$\text{GALE}\left(\boldsymbol{x},m,n,M,N,N_{L},S,\boldsymbol{p},\hat{\boldsymbol{q}},\boldsymbol{r},\boldsymbol{s},\boldsymbol{u},\boldsymbol{v},\boldsymbol{w}\right)$}

\label{algo:GALDFT-EXEC}

\normalsize{

\begin{algorithmic}[1]

\STATE{\textbf{for} $j\in\{0,1,\dots,n-1\}$ \textbf{do}\label{step:GALFD-FFT-start}}

\STATE{$\quad\quad$\textbf{for} $i\in\{0,1,\dots,m-1\}$ \textbf{do}\label{step:GALFD-first-mul-start}}

\STATE{$\quad\quad\quad\quad$$x_{i,j}\leftarrow x_{i,j}p_{i}$\label{step:GALFD-first-mult}}

\STATE{$\quad\quad$$\boldsymbol{U}^{j}\leftarrow\text{FFT}\left(\boldsymbol{x}^{j},m,M\right)$\label{step:GALFD-FFT-finish}}

\STATE{\textbf{for} $I\in\{0,1,\dots,M-1\}$ \textbf{do}}

\STATE{$\quad\quad$$\boldsymbol{V}_{I}\leftarrow\text{CZT}\left(\boldsymbol{U}_{I},n,\nicefrac{N_{L}}{2}+2(S+1),\hat{\boldsymbol{q}}^{I},\boldsymbol{r}^{I},\boldsymbol{s}^{I}\right)$\label{step:GALFD-CZT-compute}}

\STATE{$\quad\quad$\textbf{for} $K\in\{0,1,\dots,N-1\}$ \textbf{do}\label{step:GALFD-final-sum-supstart}}

\STATE{$\quad\quad\quad\quad$$Y_{I,K}\leftarrow0$\label{step:GALFD-final-sum-start}}

\STATE{$\quad\quad\quad\quad$\textbf{for} $i\in\left\{ 0,1,\dots,u_{I,K}-1\right\} $
\textbf{do}}

\STATE{$\quad\quad\quad\quad\quad\quad$$J\leftarrow v_{i}^{I,K}$}

\STATE{$\quad\quad\quad\quad\quad\quad$$Y_{I,K}\leftarrow Y_{I,K}+w_{i}^{I,K}V_{I,J}$\label{step:GALFD-final-sum-finish}}

\STATE{\textbf{return} $(\boldsymbol{Y})$}

\end{algorithmic}}
\end{algorithm}

Now we discuss the specific algorithm GALE that we are introducing
in the present paper. Algorithm~\ref{algo:GALDFT-INIT} is the initialization
phase, which is run only once for a combination of parameters $\boldsymbol{\theta}\in[\nicefrac{\pi}{4},\nicefrac{3\pi}{4})^{N}$,
$m$, $n$, $M$, $N$, $N_{L}$, $S$, $R_{1}$, $R_{2}$, $\sigma$,
and $\varepsilon$. Algorithm~\ref{algo:GALDFT-EXEC} uses the output
of Algorithm~\ref{algo:GALDFT-INIT} in order to compute an approximation
to the DFT over the $\mathfrak{R}_{M,\sigma}(\theta_{K})$ of \eqref{eq:LRFS-def}
for $K\in\{0,1,\dots,N-1\}$. Parameters $N_{L}$, $S$, and $\varepsilon$
influence the quality of the approximation. Parameters $m$, $n$,
$M$, and $N$ are related to the dimensions of the image and the
domain. Parameter $R_{1}$ determines the range $I\in\left\{ -R_{1},-R_{1}+1,\dots,M-R_{1}-1\right\} $
of the domain in order to match the definition~\eqref{eq:LRFS-def},
and parameter $R_{2}$ determines the range $J\in\left\{ -R_{2},-R_{2}+1,\dots,\nicefrac{N_{L}}{2}+2(S+1)-R_{2}-1\right\} $
of elements $Z_{I,J}$ from~\eqref{eq:ZIJ} that are going to be
computed to be used in the truncated sum~\eqref{eq:truncated-sum}.

Assuming that we have run the initialization as
\[
(\boldsymbol{p},\hat{\boldsymbol{q}},\boldsymbol{r},\boldsymbol{s},\boldsymbol{u},\boldsymbol{v},\boldsymbol{w})\leftarrow\text{GALE-I}\left(\boldsymbol{\theta},m,n,M,N,N_{L},S,R_{1},R_{2},\sigma,\varepsilon\right)
\]
with $R_{1}=\nicefrac{M}{2}-1$ and $R_{2}=\nicefrac{N_{L}}{4}+S+1$,
we analyze the output of
\[
(\boldsymbol{Y})\leftarrow\text{GALE}\left(\boldsymbol{x},m,n,M,N,N_{L},S,\boldsymbol{p},\hat{\boldsymbol{q}},\boldsymbol{r},\boldsymbol{s},\boldsymbol{u},\boldsymbol{v},\boldsymbol{w}\right).
\]
First notice that, because of the way $\boldsymbol{p}$ is initialized
in Steps~\ref{step:p-init-start}-\ref{step:p-init-finish} of Algorithm~\ref{algo:GALDFT-INIT},
Steps~\ref{step:GALFD-FFT-start}-\ref{step:GALFD-FFT-finish} of
Algorithm~\ref{algo:GALDFT-EXEC} store in the array $\boldsymbol{U}\in\mathbb{C}^{Mn}$
the elements $U_{I,J}=X_{I-R_{1},j}$, with $X_{I-R_{1},j}$ as defined
in~\eqref{eq:XIj}, for $I\in\{0,1,\dots,M-1\}$ and $j\in\{0,1,\dots,n-1\}$.
Fix an $I\in\{0,1,\dots,M-1\}$ and notice that Step~\ref{step:GALFD-mod-r}
of Algorithm~\ref{algo:GALDFT-INIT} modifies the $\boldsymbol{r}^{I}$
that was previously obtained at Step~\ref{step:GALFD-CZT-init} of
Algorithm~\ref{algo:GALDFT-INIT} in a way such that Step~\ref{step:GALFD-CZT-compute}
of Algorithm~\ref{algo:GALDFT-EXEC} computes $V_{I,J}=Z_{I-R_{1},J-R_{2}}$,
with $Z_{I-R_{1},J-R_{2}}$ as defined in~\eqref{eq:ZIJ}, for $J\in\left\{ 0,1,\dots,\nicefrac{N_{L}}{2}+2(S+1)\right\} $.
Fix a $K\in\{0,1,\dots,N-1\}$ and notice that the $u_{I,K}$ recorded
in Step~\ref{step:GALFD-u-init} of Algorithm~\ref{algo:GALDFT-INIT}
stores the number of terms in 
\[
\sum_{\left|J-\eta_{K}\right|\leq S}\frac{\hat{W}_{I}\left(\eta_{K}-J\right)}{2\pi e^{\imath\left(\eta_{K}-J\right)\varpi_{I}}}Z_{I,J}
\]
of the approximation~\eqref{eq:truncated-sum}. Furthermore, according
to Step~\ref{step:GALFD-v-init} of Algorithm~\ref{algo:GALDFT-INIT},
the set $\left\{ v_{0}^{I,K},\dots,v_{u_{I,K}-1}^{I,K}\right\} $
equals the set $\left\{ J+R_{2}:\left|J-\eta_{K}\right|\leq S\right\} $
of the indices that go into the above summation, while Step~\ref{step:GALFD-w-init}
of Algorithm~\ref{algo:GALDFT-INIT} ensures that $w_{i}^{I,K}=\frac{\hat{W}_{I}\left(\eta_{K}-J_{i}\right)}{2\pi e^{\imath\left(\eta_{K}-J_{i}\right)\varpi_{I}}}$,
where $J_{i}:=v_{i}^{I,K}-R_{2}$. Finally, Steps~\ref{step:GALFD-final-sum-start}-\ref{step:GALFD-final-sum-finish}
of Algorithm~\ref{algo:GALDFT-EXEC} make use of these pre-computed
values in order to compute
\begin{eqnarray*}
Y_{I,K} & {}=\sum_{i=0}^{u_{I,K}-1}w_{i}^{I,K}V_{I,v_{i}^{I,K}}=\sum_{|J-\eta_{K}|\le S}\frac{\hat{W}_{I}\left(\eta_{K}-J\right)}{2\pi e^{\imath\left(\eta_{K}-J\right)\varpi_{I}}}V_{I,J+R_{2}}\\
 & {}=\sum_{|J-\eta_{K}|\le S}\frac{\hat{W}_{I}\left(\eta_{K}-J\right)}{2\pi e^{\imath\left(\eta_{K}-J\right))\varpi_{I}}}Z_{I-R_{1},J}.
\end{eqnarray*}
This shows that we have $Y_{I,K}=D_{I-R_{1},N_{L},S}(\theta_{K})$,
where $D_{I-R_{1},N_{L},S}(\theta_{K})$ is as defined in~\eqref{eq:truncated-sum}.
That is, each column $\boldsymbol{Y}^{K}$ of $\boldsymbol{Y}$ returned
by Algorithm~\ref{algo:GALDFT-EXEC} contains the approximations
for $\mathcal{D}[\boldsymbol{x}]\biggl(\left(\frac{2\pi I}{M}-\sigma\right)\cot\theta_{K},\frac{2\pi I}{M}-\sigma\biggr)$
with $I\in\{-\nicefrac{M}{2}+1,-\nicefrac{M}{2}+2,\dots,\nicefrac{M}{2}\}$.

Now, let us consider the case where we have $\boldsymbol{\theta}\in[\nicefrac{3\pi}{4},\nicefrac{5\pi}{4})^{\tilde{N}}$
and we wish to compute $\mathcal{D}[\boldsymbol{x}]\biggl(\frac{2\pi I}{M}+\sigma,\left(\frac{2\pi I}{M}+\sigma\right)\tan\theta_{K}\biggr)$
for $\boldsymbol{x}\in\mathbb{R}^{mn}$ , $K\in\{0,1,\dots,\tilde{N}-1\}$,
and $I\in\left\{ \nicefrac{M}{2},-\nicefrac{M}{2}+2,\right.$ $\left.\dots,\nicefrac{M}{2}-1\right\} $.
In this case, denote $\tilde{\sigma}:=-\sigma$, let $\tilde{\boldsymbol{x}}\in\mathbb{C}^{nm}$
be such that $\tilde{x}_{i,j}=x_{j,i}$ and let $\tilde{\boldsymbol{\theta}}\in[\nicefrac{\pi}{4},\nicefrac{3\pi}{4})^{\tilde{N}}$
be given componentwise by $\tilde{\theta}_{K}=\Lambda\left(-\theta_{K}-\nicefrac{\pi}{2}\right)$.
Notice that $\tan\theta_{K}=\cot\tilde{\theta}_{K}$. Then, let us
proceed from the definition, denoting $\tilde{n}=m$ and $\tilde{m}=n$,
\begin{eqnarray*}
\fl\mathcal{D}[\boldsymbol{x}]\left(\frac{2\pi I}{M}+\sigma,\left(\frac{2\pi I}{M}+\sigma\right)\tan\theta_{K}\right) & {}=\sum_{i=0}^{m-1}\sum_{j=0}^{n-1}x_{i,j}e^{-\imath\left(j\left(\frac{2\pi I}{M}+\sigma\right)+i\left(\frac{2\pi I}{M}+\sigma\right)\tan\theta_{K}\right)}\\
 & {}=\sum_{\tilde{i}=0}^{\tilde{m}-1}\sum_{\tilde{j}=0}^{\tilde{n}-1}\tilde{x}_{\tilde{i},\tilde{j}}e^{-\imath\left(\tilde{j}\left(\frac{2\pi I}{M}-\tilde{\sigma}\right)\cot\tilde{\theta}_{K}+\tilde{i}\left(\frac{2\pi I}{M}-\tilde{\sigma}\right)\right)}\\
 & {}=\mathcal{D}[\tilde{\boldsymbol{x}}]\biggl(\left(\frac{2\pi I}{M}-\tilde{\sigma}\right)\cot\theta_{K},\frac{2\pi I}{M}-\tilde{\sigma}\biggr).
\end{eqnarray*}
Thus, we can compute the desired values by the call
\[
\left(\tilde{\boldsymbol{p}},\tilde{\hat{\boldsymbol{q}}},\tilde{\boldsymbol{r}},\tilde{\boldsymbol{s}},\tilde{\boldsymbol{u}},\tilde{\boldsymbol{v}},\tilde{\boldsymbol{w}}\right)\leftarrow\text{GALE-I}\left(\tilde{\boldsymbol{\theta}},\tilde{m},\tilde{n},M,\tilde{N},N_{L},S,\tilde{R}_{1},\tilde{R}_{2},\tilde{\sigma},\epsilon\right)
\]
with $\tilde{R}_{1}=\nicefrac{M}{2}$ and $\tilde{R}_{2}=\nicefrac{N_{L}}{4}+S$,
followed by the call
\[
\left(\tilde{\boldsymbol{Y}}\right)\leftarrow\text{GALE}\left(\tilde{\boldsymbol{x}},\tilde{m},\tilde{n},M,\tilde{N},N_{L},S,\tilde{\boldsymbol{p}},\tilde{\hat{\boldsymbol{q}}},\tilde{\boldsymbol{r}},\tilde{\boldsymbol{s}},\tilde{\boldsymbol{u}},\tilde{\boldsymbol{v}},\tilde{\boldsymbol{w}}\right),
\]
after which we have the approximations for $\mathcal{D}[\boldsymbol{x}]\biggl(\frac{2\pi I}{M}+\sigma,\left(\frac{2\pi I}{M}+\sigma\right)\tan\theta_{K}\biggr)$
for the indices $\left\{ I\in-\nicefrac{M}{2},-\nicefrac{M}{2}+1,\dots,\nicefrac{M}{2}-1\right\} $
stored in column $\tilde{\boldsymbol{Y}}^{K}$ of $\tilde{\boldsymbol{Y}}$.

We end the present subsection discussing the computation of the adjoint
operator of the approximate DFT over a GALFD that is computed by Algorithm~\ref{algo:GALDFT-EXEC}.
Notice that the algorithm can be expressed as $\boldsymbol{Y}=\mathrm{G}\boldsymbol{x}$
with $\mathrm{G}=\mathrm{L}\mathrm{C}\mathrm{F}\mathrm{D}$, where
$\mathrm{D}:\mathbb{C}^{mn}\to\mathbb{C}^{mn}$ is a diagonal linear
operator that represents Steps~\ref{step:GALFD-first-mul-start}-\ref{step:GALFD-first-mult}
of Algorithm~\ref{algo:GALDFT-EXEC} for all $j\in\{0,1,\dots,n-1\}$,
$\mathrm{F}:\mathbb{C}^{mn}\to\mathbb{C}^{Mn}$ is a linear operator
that represents Step~\ref{step:GALFD-FFT-finish} of Algorithm~\ref{algo:GALDFT-EXEC}
for all $j\in\{0,1,\dots,n-1\}$, $\mathrm{C}:\mathbb{C}^{Mn}\to\mathbb{C}^{MP}$,
where $P=\nicefrac{N_{L}}{2}+2(S+1)$, is a linear operator that represents
Step~\ref{step:GALFD-CZT-compute} of Algorithm~\ref{algo:GALDFT-EXEC}
for all $I\in\{0,1,\dots,M-1\}$ and $\mathrm{L}:\mathbb{C}^{MP}\to\mathbb{C}^{MK}$
is a linear operator that represents Steps~\ref{step:GALFD-final-sum-supstart}-\ref{step:GALFD-final-sum-finish}
of Algorithm~\ref{algo:GALDFT-EXEC}. Thus we have $\mathrm{G}^{*}=\mathrm{D^{*}}\mathrm{F}^{*}\mathrm{C}^{*}\mathrm{L}^{*}$,
where all involved adjoint operators can be computed straightforwardly,
except for $\mathrm{C}^{*}$, which we discuss soon. It is useful
to note that all of the linear operators in the adjoint computation
use the same data-independent coefficients computed by the initialization
routines of the forward operator, eliminating the need to redo any
computation when the adjoint is required, which is the case in most
applications.

To compute $\mathrm{C}^{*}$, we notice that $\mathrm{C}$ can be
written as a block-diagonal operator of the form
\[
\mathrm{C}=\left(\begin{array}{cccc}
\mathrm{B}_{0}\\
 & \mathrm{B}_{1}\\
 &  & \ddots\\
 &  &  & \mathrm{B}_{M-1}
\end{array}\right)
\]
where the $\mathrm{B}_{I}:\mathbb{C}^{n}\to\mathbb{C}^{P}$ are given
by $\mathrm{B}=\mathrm{D}_{3}^{I}\mathrm{H}_{2}\mathrm{D}_{2}^{I}\mathrm{H}_{1}\mathrm{D}_{1}^{I}$.
In this expression, the $\mathrm{D}_{1}^{I}:\mathbb{C}^{n}\to\mathbb{C}^{n}$
are diagonal linear operators that correspond to the multiplication
in Steps~\ref{step:CZT_firstfor}-\ref{step:CZT_firstfor_end} of
Algorithm~\ref{algo:CZT}, $\mathrm{H}_{1}:\mathbb{C}^{n}\to\mathbb{C}^{2P}$
is a linear operator that corresponds to the zero-padded FFT performed
in Step~\ref{step:CZT-dftp_p} of Algorithm~\ref{algo:CZT}, the
$\mathrm{D}_{2}^{I}:\mathbb{C}^{2P}\to\mathbb{C}^{2P}$ are diagonal
linear operators that correspond to the multiplication in Steps~\ref{step:CZT-conv_start}-\ref{step:CZT-second-for-finish}
of Algorithm~\ref{algo:CZT}, $\mathrm{H}_{2}:\mathbb{C}^{2P}\to\mathbb{C}^{P}$
is a linear operator that corresponds to the truncated IFFT performed
in Step~\ref{step:CZT-conv_finish} of Algorithm~\ref{algo:CZT},
and the $\mathrm{D}_{3}^{I}:\mathbb{C}^{P}\to\mathbb{C}^{P}$ are
diagonal linear operators that correspond to the multiplication in
Steps~\ref{step:CZT-final-start}-\ref{step:CZT-final-finish} of
Algorithm~\ref{algo:CZT}. The adjoints of each of the operators
$\mathrm{D}_{3}^{I}$, $\mathrm{H}_{2}$, $\mathrm{D}_{2}^{I}$, $\mathrm{H}_{1}$,
and $\mathrm{D}_{1}^{I}$ are easy to implement and we have
\[
\mathrm{C^{*}}=\left(\begin{array}{cccc}
\mathrm{B}_{0}^{*}\\
 & \mathrm{B}_{1}^{*}\\
 &  & \ddots\\
 &  &  & \mathrm{B}_{M-1}^{*}
\end{array}\right),
\]
where $(\mathrm{B}^{I})^{*}=(\mathrm{D}_{1}^{I})^{*}\mathrm{H}_{1}^{*}(\mathrm{D}_{2}^{I})^{*}\mathrm{H}_{2}^{*}(\mathrm{D}_{3}^{I})^{*}$.

\subsection{Memory Requirements and Floating Point Operations Count\label{sub:flops}}

In the present subsection we discuss memory requirements for the operation
of the proposed method and the asymptotic number of flops required
for its execution. We do the estimations for Algorithm~\ref{algo:GALDFT-EXEC}
for the case $\boldsymbol{\theta}\in[\nicefrac{\pi}{4},\nicefrac{3\pi}{4})^{N}$;
for the full range of rays, the actual memory and flops will be roughly
the double.

We start with a discussion about the memory requirements. Algorithm~\ref{algo:GALDFT-EXEC}
requires storage of the precomputed arrays $\boldsymbol{p}\in\mathbb{C}^{m}$,
$\hat{\boldsymbol{q}}\in\mathbb{C}^{2PM}$, where $P=\nicefrac{N_{L}}{2}+2(S+1)$,
$\boldsymbol{r}\in\mathbb{C}^{Mn}$, $\boldsymbol{s}\in\mathbb{C}^{MP}$,
$\boldsymbol{u}\in\mathbb{N}^{MN}$, $\boldsymbol{v}\in\mathbb{N}^{(2S+1)MN}$,
and $\boldsymbol{w}\in\mathbb{C}^{(2S+1)MN}$. It is worth pointing
out that the sizes of $\boldsymbol{u}$ and $\boldsymbol{v}$ can
be reduced to $\boldsymbol{u}\in\mathbb{N}^{N}$ and $\boldsymbol{v}\in\mathbb{N}^{(2S+1)N}$
because the indices are the same for every point of a ray $\mathfrak{R}_{M,\sigma}(\theta_{K})$,
which can be seen from the geometry of the domain. Therefore, if we
assume that the number $N$ of angles in the domain is approximately
the same as the number of rows/columns of the working images, and
because $M\ge m$, then we have to store $2S+1$ arrays of roughly
the size of the images input to the algorithm. This implies that increasing
the number of summation terms has a large impact on storage requirements,
which may be important to consider in memory-constrained computing
environments. Other than these stored precomputed data, upon execution,
the method requires two arrays $\boldsymbol{U}\in\mathbb{C}^{Mn}$
and $\boldsymbol{V}\in\mathbb{C}^{MP}$, but this extra storage can
be shared if some consideration is put in the implementation and the
amount of memory taken by these arrays is likely to be much smaller
than the one used by the linear combination coefficients $\boldsymbol{w}$.
This discussion implies that storage requirement for the method to
run is $O(SMN)$, where $S\ll m$, $M\ge m$, and $N\approx n$.

Now, we discuss the number of flops necessary for to execute Algorithm~\ref{algo:GALDFT-EXEC}.
The initialization of the precomputed coefficients will be ignored,
under the assumption that its cost is going to be diluted across many
executions of the algorithm itself. We assume that an FFT/IFFT of
length $n$ takes less than $C_{\text{FFT}}n\log_{2}n$ flops to complete,
where $n$ is the FFT/IFFT length and $C_{\text{FFT}}$ is a positive
constant. Also, we recall that complex multiplications take $6$ flops
($4$ real products and $2$ real additions) each and complex additions
use $2$ flops each. Then, notice that Steps~\ref{step:GALFD-FFT-start}-\ref{step:GALFD-FFT-finish}
of Algorithm~\ref{algo:GALDFT-EXEC} perform $6mn+nC_{\text{FFT}}M\log_{2}M$
flops. Without going into the details, the reader can realize that
each call to CZT at Step~\ref{step:GALFD-CZT-compute} of Algorithm~\ref{algo:GALDFT-EXEC}
takes $6(n+3P)+4C_{\text{FFT}}P\left(\log_{2}P+1\right)$ flops, totaling
$6M(n+3P)+4MC_{\text{FFT}}P\left(\log_{2}P+1\right)$ flops for all
the calls. It is easy to notice that $u_{I,K}\le2S+1$, therefore,
the computational effort in Steps~\ref{step:GALFD-final-sum-start}-\ref{step:GALFD-final-sum-finish}
of Algorithm~\ref{algo:GALDFT-EXEC} is less than or equal to $8N(2S+1)$,
and these steps are repeated $M$ times. The total number of flops
for Algorithm~\ref{algo:GALDFT-EXEC} therefore amounts to no more
than
\begin{eqnarray*}
6(mn+Mn+3PM)+8M & N(2S+1)\\
 & +MC_{\text{FFT}}\left[n\log_{2}n+4P\left(\log{}_{2}P+1\right)\right].
\end{eqnarray*}
Assuming, as is often the case, that $n\approx m$, $M\approx m$,
and $N\approx m$, then this simplifies to
\begin{equation}
O(m^{2}S)+O\left(mP\log_{2}P\right).\label{eq:GALAS-flops}
\end{equation}

\section{\label{sec:Numerical-Experimentation}Numerical Experimentation}

In the present section, we perform numerical experimentation intended
to provide evidence of how the algorithm works in practice. We split
the experiments in two parts. The first of these parts is dedicated
to give a feeling of how accurate and fast the method is when compared
with what we believe to be the most competitive existing alternatives
for the same abstract mathematical task. Then we use the technique
in MRI reconstruction, in order to display its practical capabilities
for real-world tasks.

\subsection{Experimentation with Simulated Data\label{sub:simulated}}

We now present experiments comparing the accuracy of GALE with alternatives
also capable of tackling the approximate computation of the DFT over
a GALFD. The results obtained by the alternative methods and those
obtained by GALE will be compared to results obtained by direct computation
of the DFT over the GALFD through the definition~\eqref{eq:DSFT_def},
in order to assess the accuracy of the approaches.

The methods we try against GALE are two: the Non-uniform Fast Fourier
Transform (NFFT)~\cite{kkp09}, which has an available implementation
in the C language in the form of a widely distributed library called
NFFT3, and the Non-Uniform Fast Fourier Transform (NUFFT)~\cite{fes03}.
Computational times are not directly comparable in the NUFFT case,
because the implementations of the GALE and the NFFT are in C/C++
while the NUFFT is entirely implemented in MATLAB (R2017b). The NFFT
algorithm bears similarity with GALE in that the approximation is
based on Theorem~\ref{thm:Fourmont}. However, in the NFFT approach,
the approximate summation is based on the DFT over a CFD, instead
of on the DFT over an LFD, as we propose. The NFFT computes the DTFT
samples through the following formula:
\[
\mathcal{D}[\boldsymbol{x}](\xi,\upsilon)\approx\sum_{\left|I-\eta_{\upsilon}\right|\leq S}\omega_{\xi,\upsilon,I}\sum_{\left|J-\eta_{\xi}\right|\leq S}\gamma_{\xi,\upsilon,J}\sum_{i=0}^{m-1}\overline{\omega}_{i}e^{-\imath i\frac{2\pi I}{P}}\sum_{j=0}^{n-1}\overline{\gamma}_{j}x_{i,j}e^{-\imath j\frac{2\pi J}{P}},
\]
where the constants $\omega_{\xi,\upsilon,I}$, $\gamma_{\xi,\upsilon,I}$,
$\overline{\omega}_{i}$, $\overline{\gamma}_{j}$, $\eta_{\upsilon}$,
and $\eta_{\xi}$ are determined as required by Theorem~\ref{thm:Fourmont}
to give an approximation to the actual DTFT values we intend to compute.
The exact values of these constants are not important for our discussion;
only the impact of the parameters $P$ and $S$, which control accuracy
and determine computational cost, matter in what follows. Both parameters
have similar meanings to their counterparts in GALE: $P$ determines
the Fourier sampling rate, while $S$ dictates the length of the truncation
of the infinite sum in Theorem~\ref{thm:Fourmont}.

Because it is based on the DFT over a CFD, the NFFT algorithm requires
oversampling and approximations in both the vertical and horizontal
directions, a fact that translates to requiring a truncated summation
with up to $(2S+1)^{2}$ terms, whereas GALE needs at most $2S+1$
terms to be summed. The NUFFT method is similar to the NFFT, and the
important point to observe here is that it too requires $O\left(S^{2}\right)$
terms in its truncated sum. On the other hand, GALE requires longer
FFTs because of the CZT computations. These longer FFTs may offer
better accuracy because of the ``natural oversampling'' phenomenon,
but this is only true within the central regions of the domain. Under
the assumptions that $n\approx m$, $M\approx m$, and $N\approx m$,
the flops count for the NFFT algorithm is
\begin{equation}
O\left(m^{2}S^{2}\right)+O\left(P^{2}\log_{2}P\right).\label{eq:NFFT-flops}
\end{equation}
This formula compares unfavorably to~\eqref{eq:GALAS-flops}, but
it is \emph{a priori} unclear which method is advantageous when the
accuracy versus computation time trade-off is taken into account.

We computed the actual values of 
\begin{equation}
\mathbb{C}^{512\cdot400}\ni\boldsymbol{y}:=\mathrm{D}_{\mathfrak{N}_{512,400}}\boldsymbol{x}^{\dagger},\label{eq:y_def}
\end{equation}
where $\boldsymbol{x}^{\dagger}\in\mathbb{R}^{512\cdot512}$ is the
synthetic image shown in Figure~\ref{fig:herman-phantom}, using~\eqref{eq:DSFT_def}
in order to have a reference for the DFT over the GALFD that incurs
no errors other than those caused by finite arithmetic precision.
After that, we computed approximations for $\mathrm{D}_{\mathfrak{N}_{512,400}}\boldsymbol{x}$
using the NFFT, the NUFFT, and the GALE with several different parameters.
The truncation parameter took values $S\in\{2,4,6,8\}$ and the Fourier
sampling rate parameter took values $P\in\{768,1024,1280\}$. Notice
that for GALE we have followed the convention $P=\nicefrac{N_{L}}{2}+2(S+1)$,
i.e., once the values for $P$ and $S$ have been fixed, GALE was
used with $N_{L}=2P-4(S+1)$. This approach was adopted in order to
make the length of the CZTs computed during the execution of GALE
equal to $P$.

\begin{figure}
\begin{centering}
{%
   \setlength{\fboxsep}{0pt}%
   \setlength{\fboxrule}{0.1pt}%
   \fbox{\includegraphics[width=0.5\columnwidth]{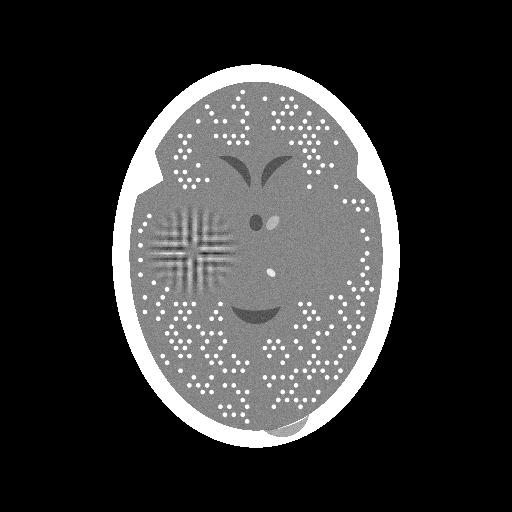}}%
}
\par\end{centering}

\caption{Mathematical phantom used in the simulated experiments. Image range
is $[0,1]$.\label{fig:herman-phantom}}
\end{figure}

A parallel version of GALE, called GALEP, was also tested to see how
the algorithm scales in multiprocessing environments. All algorithms
were executed in a quad core i7 7700HQ processor with 32GB of RAM.
GALEP used four threads of processing. GALE and GALEP were implemented
in C/C++ with the FFTW3 library~\cite{frj05} used for FFT computations.

For measures of approximation quality, we used the following: Given
a $\boldsymbol{y}\in\mathbb{C}^{MN}$ such that $y_{I,J}\ne0$ for
all $(I,J)\in\{0,1,\dots,M-1\}\times\{0,1,\dots,N-1\}$ and an approximation
$\overline{\boldsymbol{y}}$ to $\boldsymbol{y}$,
\begin{enumerate}
\item the Mean Relative Error (MRE) is defined as
\[
\text{MRE}\left(\boldsymbol{y},\overline{\boldsymbol{y}}\right):=\frac{1}{MN}\sum_{I=0}^{M-1}\sum_{J=0}^{N-1}\frac{\left|y_{I,J}-\overline{y}_{I,J}\right|}{\left|y_{I,J}\right|};
\]

\item the Relative Squared Error (RSE) is defined as 
\[
RSE\left(\boldsymbol{y},\overline{\boldsymbol{y}}\right):=\frac{\sum_{I=0}^{M-1}\sum_{J=0}^{N-1}\left|y_{I,J}-\overline{y}_{I,J}\right|^{2}}{\sum_{I=0}^{M-1}\sum_{J=0}^{N-1}\left|y_{I,J}\right|^{2}}.
\]

\end{enumerate}
\begin{figure}
\begin{centering}
\includegraphics[width=0.5\columnwidth]{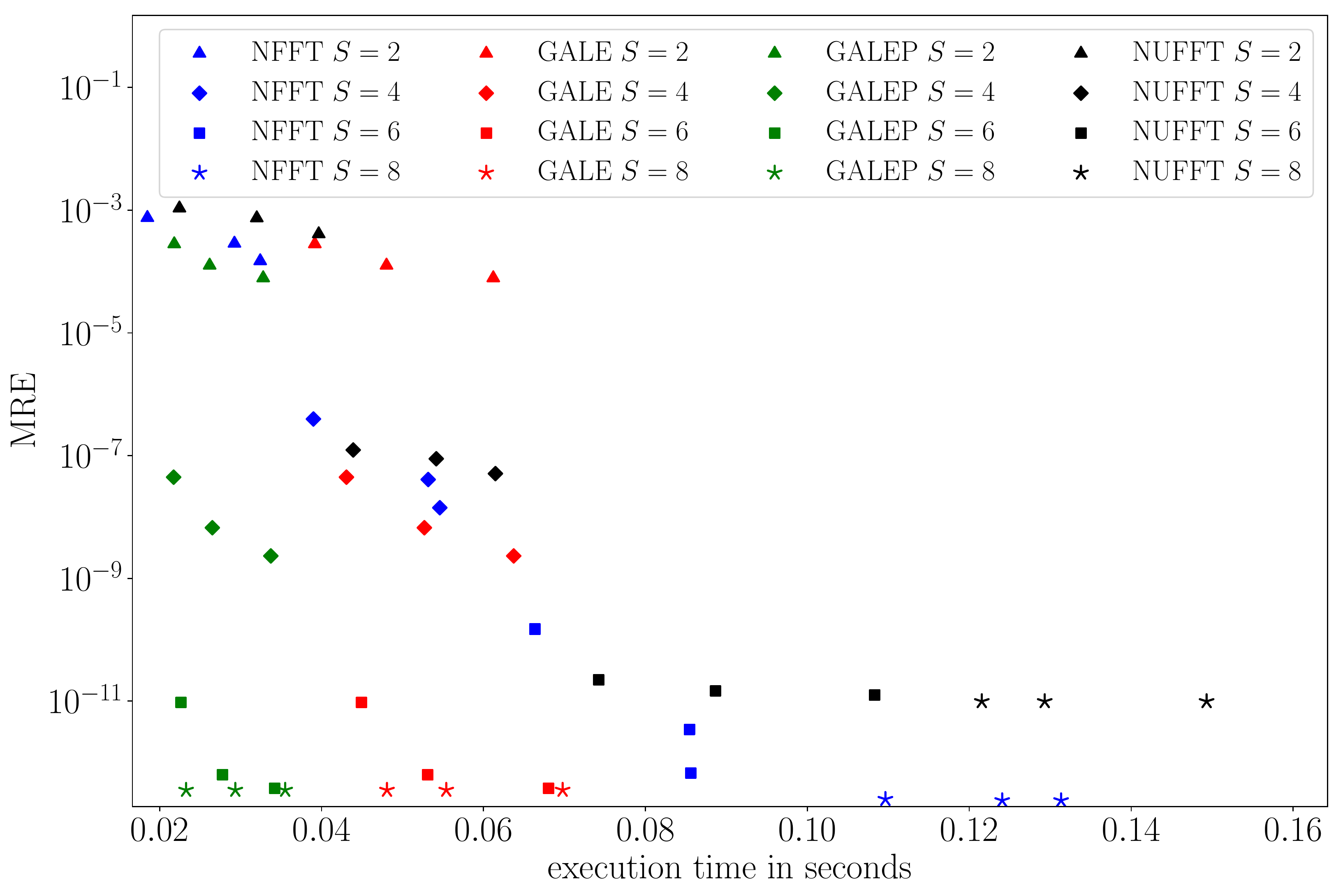}\includegraphics[width=0.5\columnwidth]{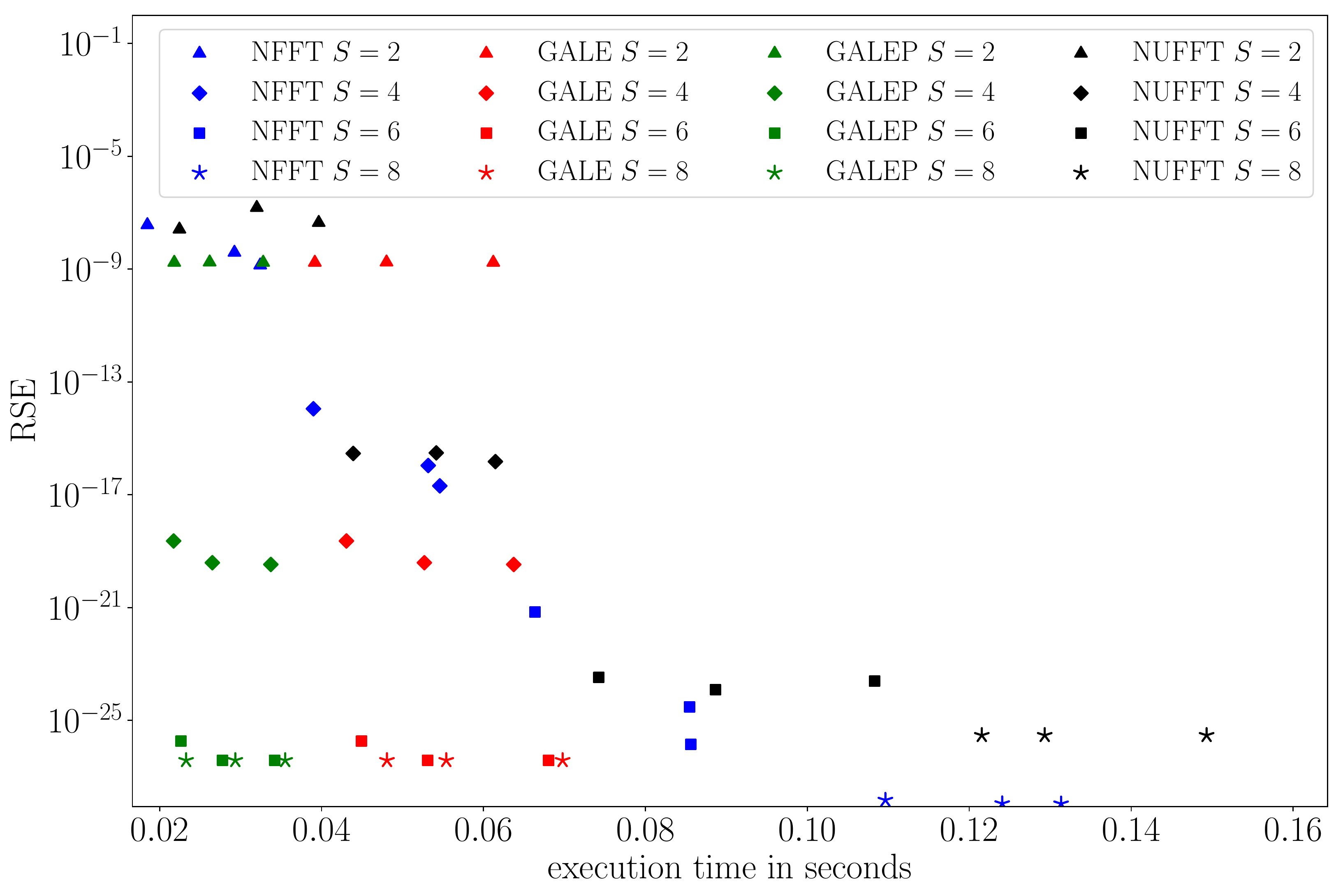}
\par\end{centering}

\caption{MREs (left) and RSEs (right) between the DFT over the GALFD computed
using~\eqref{eq:DSFT_def} and the approximation obtained by various
schemes. For a given algorithm and for a given value of the parameter
$S$, the marker for the points representing results with different
values of $P$ was kept the same, but these can be differentiated
by its computation time: the larger the $P\in\{768,1024,1280\}$,
the longer a given algorithm with a given value of $S$ takes to run
to completion.\label{fig:errors_times}}
\end{figure}

The results are shown in Figure~\ref{fig:errors_times}, where an
advantage of the GALE is seen over the NFFT and the NUFFT. It is clear
that if an MRE below $10^{-7}$ desired, then, within a same amount
of computation time, it is possible to obtain a better approximation
for the DFT over the GALFD with GALE than with the NFFT or of the
NUFFT. This is so because the $O\left(m^{2}S\right)$ term in the
flops count for GALE becomes significantly smaller than the $O\left(m^{2}S^{2}\right)$
term in the NFFT's (and NUFFT's) flops count. In fact, the $O\left(m^{2}S^{2}\right)$
term dominates~\eqref{eq:NFFT-flops} because increasing $S$ has
a larger impact on NFFT's running time than increasing $P$, whereas
the opposite is true for GALE.

Notice that the relative gains are significant. For example, for an
RSE of around $10^{-26}$, GALE takes half of the time of NFFT to
complete the computations. Furthermore, although the computation times
seem small, in MRI applications one usually deals with multiple channel
data, volumetric images, and several different contrasts in, e.g.,
quantitative imaging. This means that computation times add up and
therefore efficient numerical techniques for the computation of the
DTFT over radial domains can have important practical consequences,
especially when iterative reconstruction methods, as discussed in
\ref{sec:Iterative-Algorithms}, are used to recover the $\boldsymbol{x}$
in \eqref{eq:DSFT_def} from the observed data, indicated by the left-hand-side
of \eqref{eq:DSFT_def}.

We also notice that GALE scales well in multiprocessing environments.
A parallel CPU implementation of the NFFT (but not of the NUFFT) is
publicly available and recent research indicates that careful domain-specific
implementations of the NUFFT can be competitive in GPU environments~\cite{sss18},
however, extensive software comparisons in parallel environments are
out of the scope of the present paper and we do not further pursue
this topic here. A worth-mentioning alternative that could be adapted
to the sampling circumstances that we are considering can be found
in~\cite{pos01}.

\begin{figure}
\begin{centering}
\includegraphics[width=0.5\columnwidth]{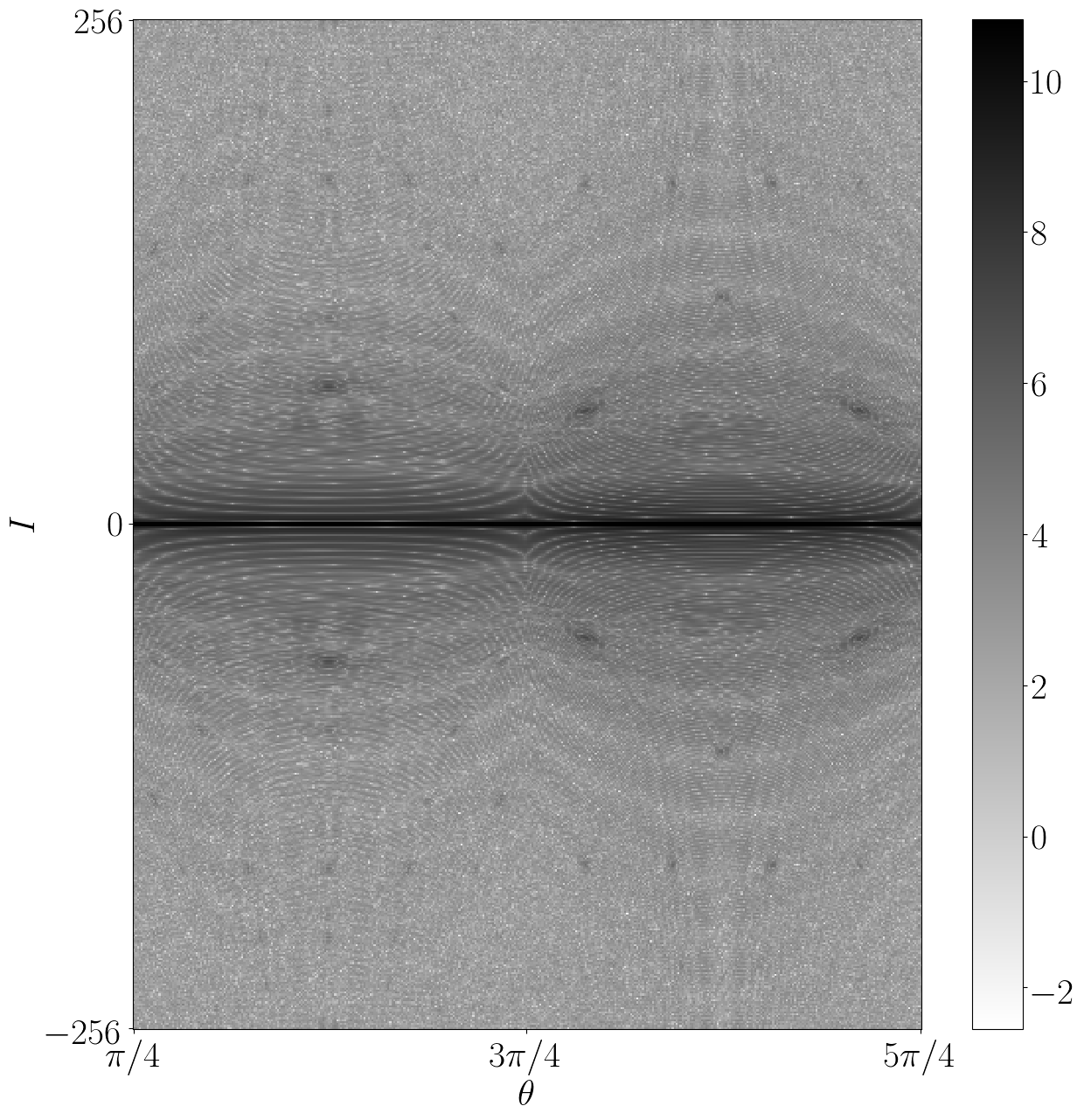}\includegraphics[width=0.5\columnwidth]{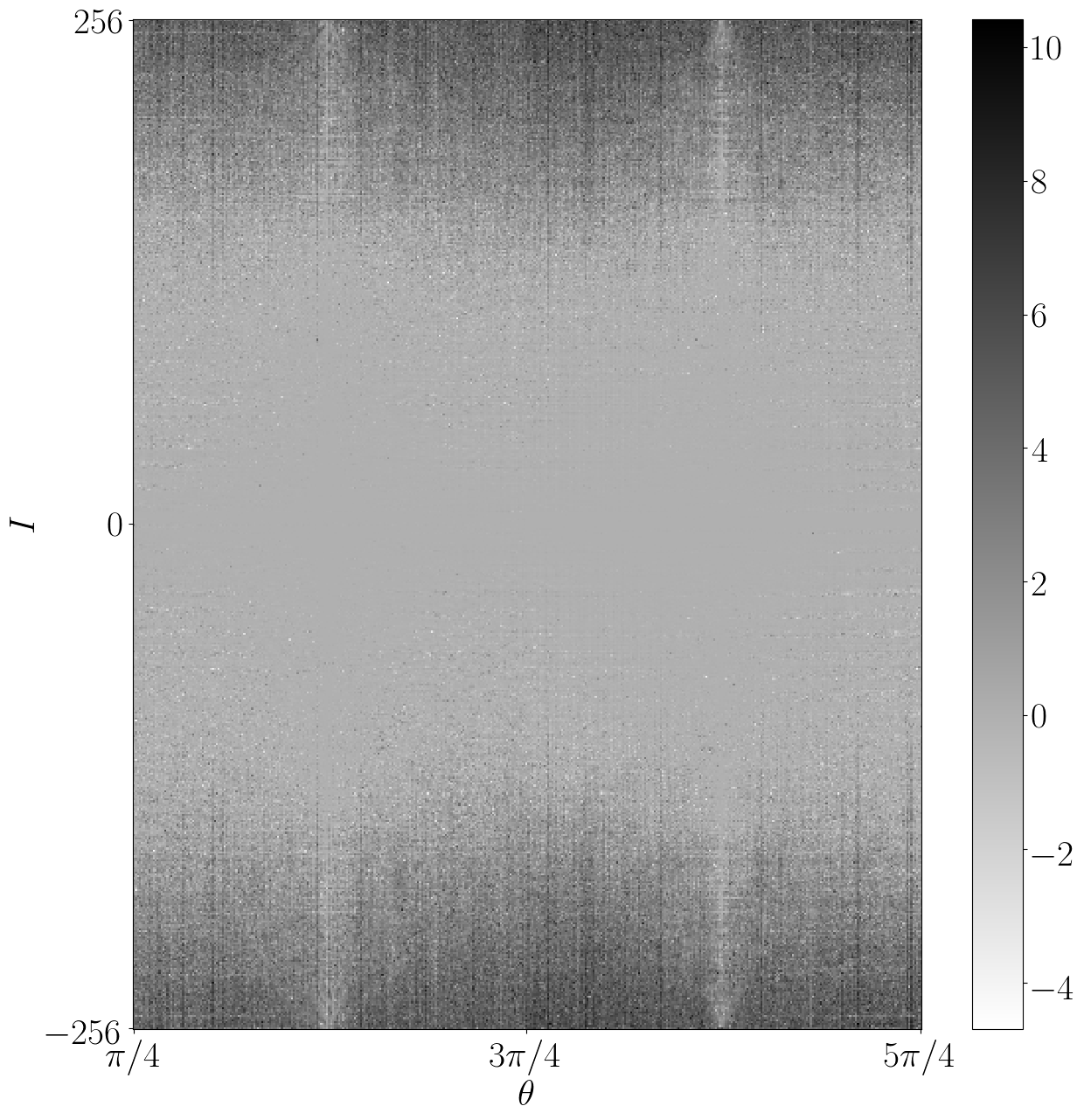}
\par\end{centering}

\caption{Left: Image of $\log\left|y_{I,J}\right|$ where $\boldsymbol{y}=\mathcal{D}_{\mathfrak{N}_{500,400}}\boldsymbol{x}$.
It can be seen that the largest values are highly concentrated around
the region $|I|\approx0$. Right: Image of $\log\frac{\left|\tilde{y}_{I,J}-y_{I,J}\right|}{\left|\overline{y}_{I,J}-y_{I,J}\right|}$,
where $\tilde{\boldsymbol{y}}$ is the approximation obtained by GALE
with $(P,S)=(768,6)$ and $\overline{\boldsymbol{y}}$ is the approximation
obtained by GALE with $(P,S)=(1280,6)$. Note that the accuracy improvements
obtained by a larger Fourier oversampling occur mainly in the region
$|I|\gg0$. In both images, each column represents one ray of the
GALFD.\label{fig:images-energy-and-ratio}}
\end{figure}

An interesting observation regarding these experiments is that if
absolute, rather than relative, errors are considered, then increasing
the Fourier oversampling parameter $P$ has little effect on the approximation
error for GALE. The explanation is given by the combination of two
reasons: \textbf{(i)} the nature of the approximation error bound~\eqref{eq:error-bound},
which, as depicted in Figure~\ref{fig:error-bound}, is such that,
for GALE, the Fourier oversampling is only important in the higher
frequencies components of the Fourier domain and \textbf{(ii)} the
fact that medical images in general, and our phantom in particular,
have most of their energy concentrated in the low frequency parts
of the Fourier domain. Figure~\ref{fig:images-energy-and-ratio}
illustrates these two points. On the left, we can see that indeed
the largest values of the DTFT of the phantom are very concentrated
in the region $|I|\approx0$. On the right, we observe that increasing
the Fourier sampling parameter does provide significantly better accuracy,
but this improvement happens mostly in the region $|I|\gg0$.

Realizing that the NUFFT does not achieve the same level of accuracy
as the other two techniques, we have devised an experiment to further
investigate this phenomenon. We have considered the function
\[
f(\boldsymbol{x})=\left\Vert \mathrm{D}_{\mathfrak{N}_{512,400}}\boldsymbol{x}-\boldsymbol{y}\right\Vert ^{2},
\]
where, again, $\boldsymbol{y}$ is as given in~\eqref{eq:y_def}
and is computed directly from the definition~\eqref{eq:DSFT_def}.
We have then run $20$ iterations of the Conjugate Gradient (CG) method
for the minimization of this function. We have performed two different
reconstructions, one using GALE and the other using the NUFFT with
parameters $S=2$ for both methods and $P=768$ for the NUFFT and
$P=520$ for the GALE to approximate $\mathrm{D}_{\mathfrak{N}_{512,400}}$
and its adjoint. The parameter values were such that the computation
time was approximately the same for both algorithms.

\begin{figure}
\includegraphics[width=0.498\columnwidth]{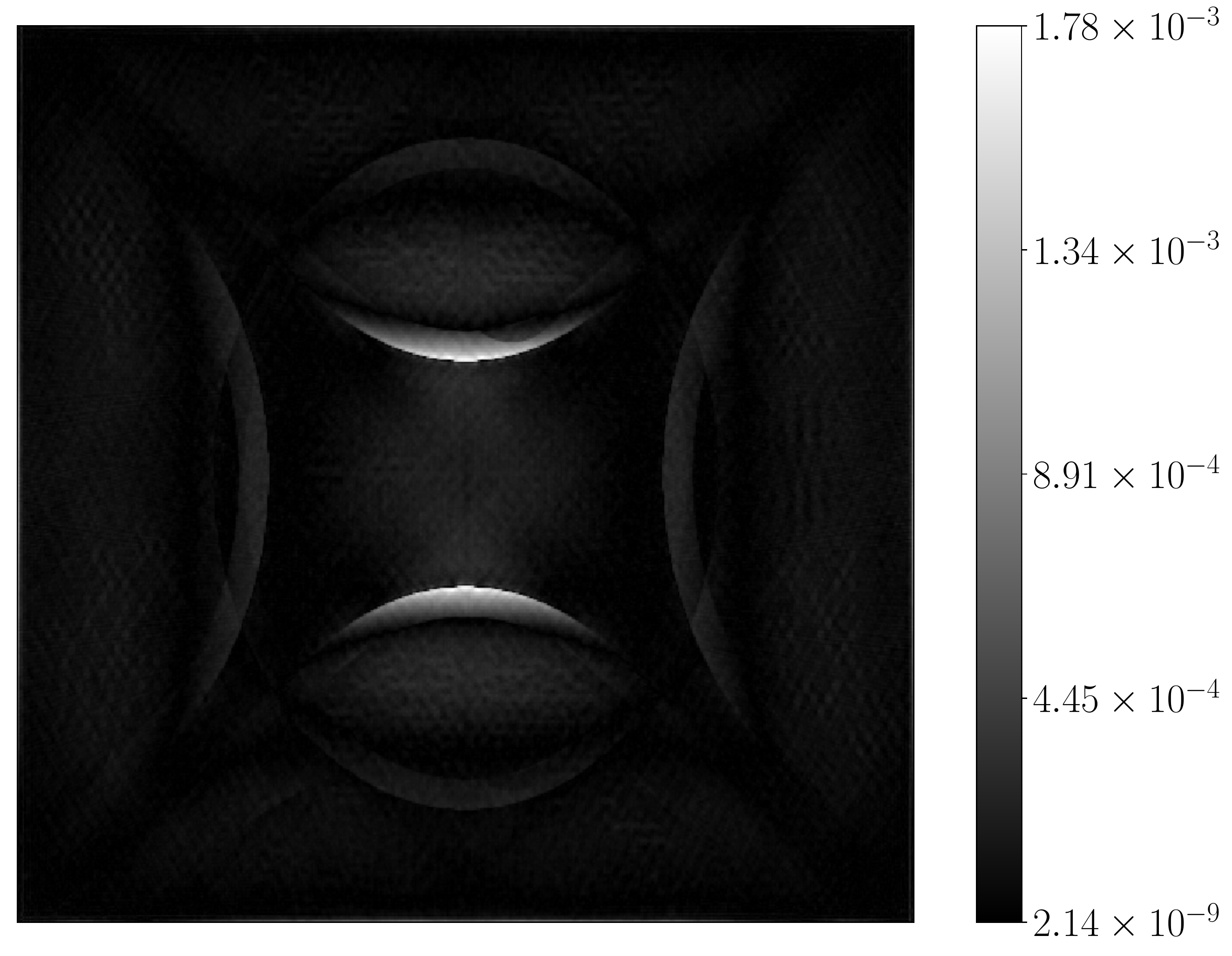}\includegraphics[width=0.502\columnwidth]{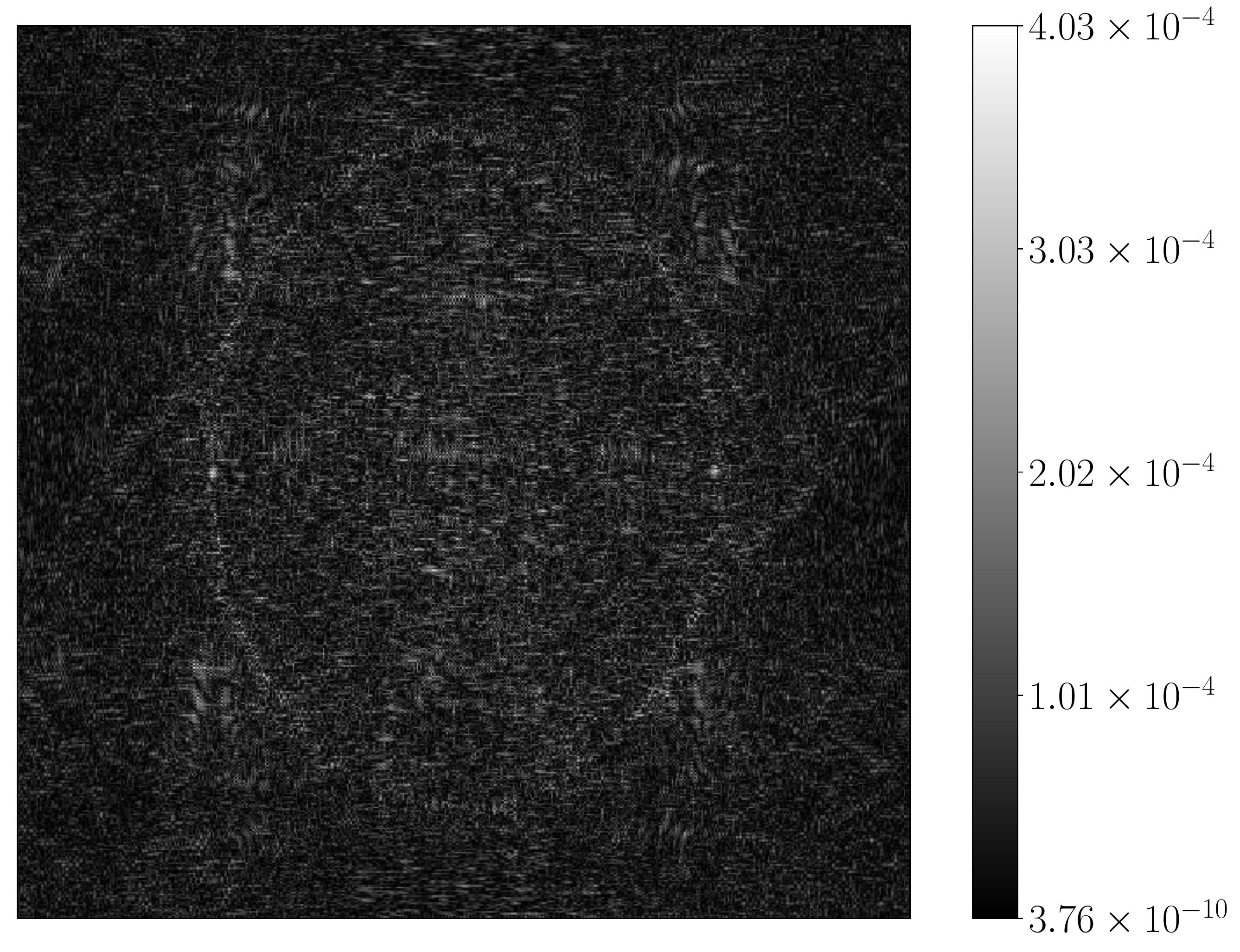}\caption{Images of the errors $\left|x_{i,j}^{\dagger}-x_{i,j}^{20}\right|$,
where $\boldsymbol{x}^{20}$ is iteration $20$ of the CG method started
with $\boldsymbol{x}^{0}=\boldsymbol{x}^{\dagger}$ as detailed in
Subsection~\ref{sub:simulated}. Left: Errors obtained using the
NUFFT with $S=2$ and $P=768$ for approximating $\mathrm{D}_{\mathfrak{N}_{512,400}}$
and its adjoint. The maximum error is $1.8\cdot10^{-3}$. Right: Errors
obtained using the GALE with $S=2$ and $P=520$ for approximating
$\mathrm{D}_{\mathfrak{N}_{512,400}}$and its adjoint. The maximum
error is $4.0\cdot10^{-4}$.\label{fig:error_images}}
\end{figure}

Differences between the reconstructed images and the image $\boldsymbol{x}^{\dagger}$
that originated the data are possibly caused by the following:
\begin{enumerate}
\item Finite termination of the iterative method;
\item Infinite number of solutions for the problem due to number of equations
smaller than the number of variables;
\item Approximation errors of the methods.
\end{enumerate}
Notice that although the system matrix in this case might not be positive
definite, only the last item will be different between experiments
using GALE and experiments using the NUFFT, so, in order to be able
to assess the influence of the last item alone, the CG algorithm was
started with the actual desired image $\boldsymbol{x}^{\dagger}$
and therefore the method should not move from the starting point if
computations were performed with full numerical accuracy. The resulting
images of the differences between reconstructed and original images
can be seen in Figure~\ref{fig:error_images} and show an interesting
phenomenon: the image obtained by the NUFFT presents a very structured
error, whereas the image obtained using the GALE has a random-like
error. In unreported experiments, we have increased the Fourier oversampling
and the truncation parameter for the NUFFT, and we were able to rule
out circular interference as the sole cause of the observed structure.
Although the magnitude of the error can indeed be reduced this way,
it always presents some sort of structure that follows the reconstructed
image.

\subsection{Magnetic Resonance Image Reconstruction}

This subsection is dedicated to experiments involving image reconstruction
in magnetic resonance. In MRI, the data acquisition process provides
a signal that corresponds to the CFT of the object, modulated by the
receiver sensitivities. More precisely, the image is assumed to be
an absolutely integrable complex-valued function in the plane $f:\mathbb{R}^{2}\to\mathbb{C}$
and the data are assumed to take the form
\begin{eqnarray*}
\fl y_{I,J}^{c}\approx\widehat{\zeta_{c}f}(\xi_{I,J},\upsilon_{I,J} & )=\int_{\mathbb{R}^{2}}\zeta_{c}(\boldsymbol{z})f(\boldsymbol{z})e^{-\imath\left(z_{1}\xi_{I,J}+z_{2}\upsilon_{I,J}\right)}\mathrm{d}\boldsymbol{z},\\
 & (I,J,c)\in\{0,1,\dots,M-1\}\times\{0,1,\dots,N-1\}\times\{0,1,\dots,C-1\},
\end{eqnarray*}
where each bounded function $\zeta_{c}:\mathbb{R}^{2}\to\mathbb{C}$
is the sensitivity function of receive coil $c$. In our case, we
acquired data with 
\[
\left\{ \left(\xi_{I,J},\upsilon_{I,J}\right):(I,J)\in\{0,1,\dots,M-1\}\times\{0,1,\dots,N-1\}\right\} =\mathfrak{N}_{M,N}.
\]
Moreover, we can assume that
\[
f=\sqrt{\sum_{c=0}^{C-1}\left(\zeta_{c}f\right)^{2}}.
\]

The radial nature of the GALFD allows us to take an interesting analytical
approach to image reconstruction from this kind of dataset, using
a tool from CT~\cite{her09,nat86,naw01}, which is the Filtered Backprojection
(FBP) algorithm. Let us assume that $f\in\mathscr{S}(\mathbb{R}^{2})$,
where $\mathscr{S}(\mathbb{R}^{2})$ is the Schwartz space~\cite{nat86},
then
\[
f=\mathcal{B}\circ\mathcal{W}\circ\mathcal{R}[f],
\]
where $\mathcal{R}$ is the Radon Transform (RT)
\[
\mathcal{R}[f](\theta,t):=\int_{\mathbb{R}}f\left(s\left(-\sin\theta,\cos\theta\right)+t\left(\cos\theta,\sin\theta\right)\right)\mathrm{d}s,
\]
$\mathcal{B}$ is the Backprojection operator
\[
\mathcal{B}[g](\boldsymbol{z}):=\int_{[0,\pi]}g(\theta,z_{1}\cos\theta+z_{2}\sin\theta)\mathrm{d}\theta,
\]
which turns out to be the adjoint of the Radon Transform $\mathcal{B}=\mathcal{R}^{*}$,
and $\mathcal{W}$ has the form
\[
\mathcal{W}=\mathcal{E}^{-1}\circ\mathcal{C}\circ\mathcal{E},
\]
where $\mathcal{E}$ is an operator that computes the one-dimensional
Fourier Transform with relation to the second variable of a function
of two variables
\[
\mathcal{E}[g](\theta,\omega):=\int_{\mathbb{R}}g(\theta,t)e^{-\imath\omega t}\mathrm{d}t,
\]
which in turn means that $\mathcal{E}^{-1}=(2\pi)^{-1}\mathcal{E}^{*}$,
and $\mathcal{C}$ is a diagonal operator of the form 
\[
\mathcal{C}[g](\theta,\omega):=|\omega|g(\theta,\omega).
\]
Therefore, denoting $\mathcal{G}:=\mathcal{E}\circ\mathcal{R}$ we
have
\begin{equation}
f=\frac{1}{2\pi}\mathcal{G}^{*}\circ\mathcal{C}\circ\mathcal{G}[f].\label{eq:inversion-continuous}
\end{equation}
(For other interesting and relevant inversion techniques see~\cite{grd11,gri13,gri16}.)
Notice that operator $\mathcal{C}$ can be described as pointwise
multiplication by a Density Compensation Function (DCF) $h(\omega,\theta)=|\omega|$
and that, because of the radial nature of the GALFD, we can write
$(\xi_{I,j},\upsilon_{I,J})=\left(\omega_{I}\cos\theta_{J},\omega_{I}\sin\theta_{J}\right)$,
where $\left|\omega_{I}\right|=\sqrt{\xi_{I,J}^{2}+\upsilon_{I,J}^{2}}$.
Now, we consider the Fourier Slice Theorem (FST)~\cite{her09,nat86,naw01},
which states that
\[
\mathcal{G}[f](\theta,\omega)=\mathcal{F}[f](\omega\cos\theta,\omega\sin\theta),
\]
and then we approximate the above integral model by a discrete one.

Assume that $\boldsymbol{x}^{c}\in\mathbb{C}^{mn}$ is such that $x_{i,j}^{c}=f_{c}(j-\nicefrac{n}{2},i-\nicefrac{m}{2})$,
where we denoted $f_{c}:=\zeta_{c}f$, then we have, using the FST
and then approximating the integrations by a sum,
\begin{eqnarray*}
y_{I,J}^{c} & {}=\mathcal{G}[f_{c}]\left(\theta_{I},\omega_{J}\right)=\mathcal{F}\left[f_{c}\right]\left(\omega_{I}\cos\theta_{J},\omega_{I}\sin\theta_{J}\right)\\
 & {}\approx\frac{1}{mn}\sum_{i=0}^{m-1}\sum_{j=0}^{n-1}x_{i,j}^{c}e^{-\imath\left((j-\nicefrac{n}{2})\xi_{I,J}+(i-\nicefrac{m}{2})\upsilon_{I,J}\right)}\\
 & {}=\frac{e^{\imath\left((\nicefrac{n}{2})\xi_{I,J}+(\nicefrac{m}{2})\upsilon_{I,J}\right)}}{mn}\sum_{i=0}^{m-1}\sum_{j=0}^{n-1}x_{i,j}^{c}e^{-\imath\left(j\xi_{I,J}+i\upsilon_{I,J}\right)}.
\end{eqnarray*}
That is, 
\[
\boldsymbol{y}^{c}\approx\mathrm{Z}\mathrm{D}_{\mathfrak{N}_{M,N}}\boldsymbol{x}^{c},
\]
where the diagonal linear operator $\mathrm{Z}:\mathbb{C}^{MN}\to\mathbb{C}^{MN}$
is given componentwise by
\[
(\mathrm{Z}\boldsymbol{y})_{I,J}=y_{I,J}\frac{e^{\imath\left((\nicefrac{n}{2})\xi_{I,J}+(\nicefrac{m}{2})\upsilon_{I,J}\right)}}{mn}.
\]
In other words, the composite operator $\mathrm{Z}\mathrm{D}_{\mathfrak{N}_{M,N}}$
serves as a discrete approximation for $\mathcal{G}[f_{c}]$ when
the latter is computed over $(\xi,\theta)\in\mathfrak{N}_{M,N}$.

The validity of the discretization discussed above depends on bandwidth
considerations about the functions $f_{c}$. An in-depth discussion
about this subject can be found in~\cite{nat86}, but for us, the
practical value of this approach is that it allows us to make use
of a discrete version of~\eqref{eq:inversion-continuous} that reads
\[
\boldsymbol{x}^{c}\approx\mathrm{D}_{\mathfrak{N}_{M,N}}^{*}\mathrm{Z^{*}}\mathrm{C}\mathrm{Z}\mathrm{D}_{\mathfrak{N}_{M,N}}\boldsymbol{x}^{c},
\]
where the diagonal linear operator $\mathrm{C}:\mathbb{C}^{MN}\to\mathbb{C}^{MN}$
is given componentwise by
\[
(\mathrm{C}\boldsymbol{y})_{I,J}=y_{I,J}\sqrt{\xi_{I,J}^{2}+\upsilon_{I,J}^{2}}.
\]
Finally, the algorithm for inverting the MR data consists of computing
\[
\boldsymbol{x}^{c}=\mathrm{D}_{\mathfrak{N}_{M,N}}^{*}\mathrm{Z^{*}}\mathrm{C}\boldsymbol{y}^{c},\quad c\in\{0,1,\dots,C-1\}
\]
and then
\[
x_{i,j}=\sqrt{\sum_{c=0}^{C-1}\left(x_{i,j}^{c}\right)^{2}},\quad(i,j)\in\{0,1,\dots,m-1\}\times\{0,1,\dots,n-1\}.
\]

The above algorithm was applied to the MRI reconstruction of a physical
phantom and of a human brain, where image dimensions were $m=n=512$
($256$ pixels with twofold oversampling, i.e., data were acquired
for an ROI twice as large as the one that was known to contain the
object in order to avoid circular interference) and data dimensions
were $M=512$, $N=400$ and $C=20$ for the phantom, and $N=800$
and $C=16$ for the brain (recall that the data were collected in
$M$ samples for each of the $N$ rays and each sample is measured
in parallel by $C$ coils). Two-dimensional Fourier data were obtained
by applying the IFFT along the $z$ direction of MRI data acquired
using a stack-of-stars 3D Gradient-Recalled Echo (GRE) linogram sequence
with golden-angle ordering at $3$ Tesla. We used two different ways
of computing the adjoint operator $\mathrm{D}_{\mathfrak{N}_{M,N}}^{*}$.
Besides GALE, the NUFFT~\cite{fes03} adjoint routine was tested
for comparison. We used $S=3$ and $P=768$ for both the NUFFT and
GALE. GALE was used with $8$ threads of execution. For these experiments,
a MATLAB binding was created for GALE in order to have it running
in the same environment as the NUFFT package, which is implemented
in MATLAB.

\begin{figure}
\begin{centering}
\def\relsize{0.4955}%
{%
   \setlength{\fboxsep}{0pt}%
   \setlength{\fboxrule}{0.1pt}%
   \fbox{\includegraphics[width=\relsize\columnwidth]{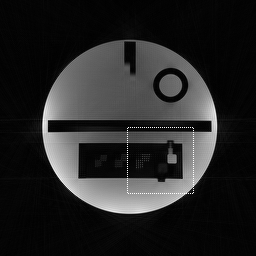}}%
}\hfill%
{%
   \setlength{\fboxsep}{0pt}%
   \setlength{\fboxrule}{0.1pt}%
   \fbox{\includegraphics[width=\relsize\columnwidth]{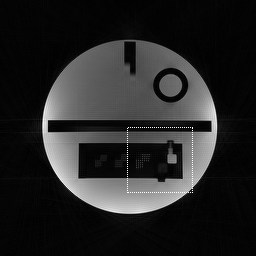}}%
}\\[0.007\columnwidth]%
{%
   \setlength{\fboxsep}{0pt}%
   \setlength{\fboxrule}{0.1pt}%
   \fbox{\includegraphics[width=\relsize\columnwidth]{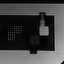}}%
}\hfill%
{%
   \setlength{\fboxsep}{0pt}%
   \setlength{\fboxrule}{0.1pt}%
   \fbox{\includegraphics[width=\relsize\columnwidth]{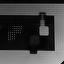}}%
}%
\par\end{centering}

\caption{Reconstruction of physical phantom image from magnetic resonance data.
Left: Using GALE ($0.025$ seconds of computation per coil per slice).
Right: Using NUFFT within ($0.034$ seconds of computation per coil
per slice). Top: Full image reconstructions. Bottom: Enlarged detail
of small reconstructed structures (from regions delimited by dotted
lines in top row).\label{fig:grid_recons}}
\end{figure}

\begin{figure}
\begin{centering}
\def\relsize{0.4955}%
{%
   \setlength{\fboxsep}{0pt}%
   \setlength{\fboxrule}{0.1pt}%
   \fbox{\includegraphics[width=\relsize\columnwidth]{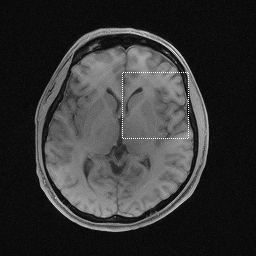}}%
}\hfill%
{%
   \setlength{\fboxsep}{0pt}%
   \setlength{\fboxrule}{0.1pt}%
   \fbox{\includegraphics[width=\relsize\columnwidth]{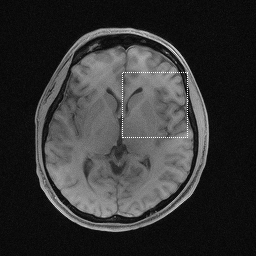}}%
}\\[0.007\columnwidth]%
{%
   \setlength{\fboxsep}{0pt}%
   \setlength{\fboxrule}{0.1pt}%
   \fbox{\includegraphics[width=\relsize\columnwidth]{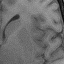}}%
}\hfill%
{%
   \setlength{\fboxsep}{0pt}%
   \setlength{\fboxrule}{0.1pt}%
   \fbox{\includegraphics[width=\relsize\columnwidth]{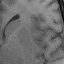}}%
}%
\par\end{centering}

\caption{Reconstruction of human brain image from magnetic resonance data.
Left: Using GALE ($0.028$ seconds of computation per coil per slice).
Right: Using NUFFT ($0.063$ seconds of computation per coil per slice).
Top: Full image reconstructions. Bottom: Enlarged detail of small
reconstructed structures (from regions delimited by dotted lines in
top row).\label{fig:grid_recons_brain}}
\end{figure}

The resulting images from the GALE and from the NUFFT are visually
indistinguishable in the two experiments, as can be seen in Figures~\ref{fig:grid_recons}~and~\ref{fig:grid_recons_brain},
with the existing numerical differences being irrelevant for human
visualization. Furthermore, the good image quality indicates that
the GALFD is a sound way of sampling the Fourier space in MRI and
that the proposed numerical approach is useful in practical scenarios,
potentially leading to a significant reduction in computation time
without loss of accuracy when compared to other available algorithms
for the same task.

\section{\label{sec:Conclusions}Conclusions}

We introduced a new family of subsets, each of which is called a GALFD,
of the two-dimensional Fourier space, that has desirable properties
for MRI reconstruction applications. The paper was dedicated to the
study of GALE, a new approach for the numerical computation of the
DFT over finite domains, in particular over a GALFD. GALE was demonstrated
to be fast and accurate when compared to existing techniques available
for the task. Therefore, among the desirable practical features of
GALE is its good accuracy/computation time trade-off. An upper bound
on the approximation error incurred by GALE was provided; it can be
used to gain insight on how to select the method's parameters in practical
settings.

Generalizing the GALFD sampling scheme to a 3D configuration is planned
for future work, again with application to MRI as the main goal~\cite{hra92}.
In this case, the computation time reductions of a linogram-based
approach is even more significant, because an appropriate selection
of the rays in 3D space allows the method to use a summation of $O(S)$
terms, while some current methods use $O\left(S^{3}\right)$ terms
for this case, which in practice precludes the use of three-dimensional
radial sampling of the Fourier space.

\appendix

\section{The 1D Discrete Fourier Transform\label{sec:FFT}}

The 1D DFT of a vector $\boldsymbol{x}\in\mathbb{C}^{m}$ is given
by:
\[
\hat{x}_{I}:=\sum_{i=0}^{m-1}x_{i}e^{-\imath i\frac{2\pi I}{m}},\quad I\in\{0,1,\dots,m-1\}.
\]
This defines a linear operator $\mathrm{F}_{m}:\mathbb{C}^{m}\rightarrow\mathbb{C}^{m}$
such that $\hat{\boldsymbol{x}}=\mathrm{F}_{m}\boldsymbol{x}$ is
given componentwise as above. It is possible to perform the computation
of $\mathrm{F}_{m}\boldsymbol{x}$ efficiently with an FFT of length
$m$ that uses roughly $5m\log_{2}m$ flops. Inversion of the 1D DFT
is simple: 
\[
x_{i}=\frac{1}{m}\sum_{I=0}^{m-1}\hat{x}_{I}e^{\imath I\frac{2\pi i}{m}},\quad i\in\{0,1,\dots,m-1\}.
\]
Such a computation is performed efficiently by the Inverse FFT (IFFT)
algorithm that uses about as many flops as the FFT does. We denote
this inverse operator by $\mathrm{F}_{m}^{-1}:\mathbb{C}^{m}\rightarrow\mathbb{C}^{m}$.

There are times when we wish to compute, for some $M\ge m$, 
\begin{equation}
\hat{x}_{I}:=\sum_{i=0}^{m-1}x_{i}e^{-\imath i\frac{2\pi I}{M}},\quad I\in\{0,1,\dots,M-1\}.\label{eq:DFT1D}
\end{equation}
Such a computation provides a finer sampling of the Fourier space
and can be obtained by computing $\mathrm{F}_{M}\mathrm{Z}_{m,M}\boldsymbol{x}$,
where $\mathrm{Z}_{m,M}:\mathbb{C}^{m}\rightarrow\mathbb{C}^{M}$
is the zero-padding operator given by
\[
(\mathrm{Z}_{m,M}\boldsymbol{x})_{i}:=\left\{ \begin{array}{ll}
x_{i}, & \quad\text{if}\quad0\le i<m-1,\\
0, & \quad\text{if}\quad m\leq i<M.
\end{array}\right.
\]
It is helpful to use the notation $\text{FFT}(\boldsymbol{x},m,M):=\mathrm{F}_{M}\mathrm{Z}_{m,M}\boldsymbol{x}$,
which also implies that such a computation is done using the FFT algorithm
of length $M$ after zero-padding the vector $\boldsymbol{x}\in\mathbb{R}^{m}$,
where $m\le M$. Correspondingly, for $\boldsymbol{x}\in\mathbb{C}^{M}$
with $M\ge m$, we denote $\text{IFFT}(\boldsymbol{x},m,M):=\mathrm{T}_{m,M}\mathrm{F}_{M}^{-1}\boldsymbol{x}$,
where $\mathrm{T}_{m,M}:\mathbb{C}^{M}\rightarrow\mathbb{C}^{m}$
is the truncation operator defined componentwise as $\left(\mathrm{T}_{m,M}\boldsymbol{x}\right)_{i}=x_{i}$,
$i\in\{0,1,\dots m-1\}$. Note that $\mathrm{T}_{m,M}\mathrm{F}_{M}^{-1}=\nicefrac{1}{M}(\mathrm{F}_{M}\mathrm{Z}_{m,M})^{*}$.

Also useful for us is to notice that the range $I\in\{0,1,\dots,M-1\}$
in~\eqref{eq:DFT1D} can be changed very easily to, say, $I\in\{-R,-R+1,\dots,-R+M-1\}$
by computing
\[
\sum_{i=0}^{m-1}x_{i}e^{-\imath i\frac{2\pi(I-R)}{M}}=\sum_{i=0}^{m-1}x_{i}e^{\imath i\frac{2\pi R}{M}}e^{-\imath i\frac{2\pi I}{M}},\quad I\in\{0,1,\dots,M-1\},
\]
that is multiplying each component $x_{i}$ of $\boldsymbol{x}\in\mathbb{R}^{m}$
by $e^{\imath i\frac{2\pi R}{M}}$ before computing $\text{FFT}(\boldsymbol{x},m,M)$.

\section{The Chirp-Z Transform\label{sub:CZT}}

Let $\alpha\in\mathbb{R}$, $R\in\mathbb{Z}$, and $m$, $M$, $P$
be positive integers such that $P\geq m$. Consider the linear operator
$\mathrm{F}_{m}^{\alpha,M,P,R}:\mathbb{C}^{m}\rightarrow\mathbb{C}^{P}$
such that, for $\boldsymbol{x}\in\mathbb{C}^{m}$, if $\hat{\boldsymbol{x}}^{\alpha,M,P,R}=\mathrm{F}_{m}^{\alpha,M,P,R}\boldsymbol{x}$,
then $\hat{\boldsymbol{x}}^{\alpha,M,P,R}\in\mathbb{C}^{P}$ is given
componentwise by 
\begin{equation}
\fl\hat{x}_{I}^{\alpha,M,P,R}:=\sum_{i=0}^{m-1}x_{i}e^{-\imath i\frac{2\pi(I-R)}{M}\alpha}=\sum_{i=0}^{m-1}x_{i}e^{\imath i\frac{2\pi R}{M}\alpha}e^{-\imath i\frac{2\pi I}{M}\alpha},\quad I\in\{0,1,\dots,P-1\}.\label{eq:CZT-def}
\end{equation}
We call the $\hat{\boldsymbol{x}}^{\alpha,M,P,R}\in\mathbb{C}^{P}$
thus obtained the Chirp-Z Transform of $\boldsymbol{x}\in\mathbb{C}^{m}$.
There is an efficient equivalent way of performing this computation,
for which a key observation is that
\begin{eqnarray}
\hat{x}_{I}^{\alpha,M,P,R} & {}=\sum_{i=0}^{m-1}x_{i}e^{\imath i\frac{2\pi R}{M}\alpha}e^{-\imath\frac{\pi\alpha}{M}\left(I^{2}+i^{2}-(I-i)^{2}\right)}\nonumber \\
 & {}=e^{-\imath\frac{\pi\alpha}{M}I^{2}}\sum_{i=0}^{m-1}x_{i}e^{\imath i\frac{2\pi R}{M}\alpha}e^{-\imath\frac{\pi\alpha}{M}i^{2}}e^{\imath\frac{\pi\alpha}{M}(I-i)^{2}}.\label{eq:CZT-CC}
\end{eqnarray}
We will return to this computation shortly; meanwhile we consider
another useful concept.

A circular convolution $\boldsymbol{p}\circledast\boldsymbol{q}\in\mathbb{C}^{Q}$
between two vectors $\boldsymbol{p}\in\mathbb{C}^{Q}$ and $\boldsymbol{q}\in\mathbb{C}^{Q}$
is defined by
\begin{equation}
(\boldsymbol{p}\circledast\boldsymbol{q})_{I}:=\sum_{i=0}^{Q-1}p_{i}q_{(I-i)\%Q},\quad I\in\{0,1,\dots,Q-1\},\label{eq:CC}
\end{equation}
 The easily proven Circular Convolution Theorem (CCT) states that
\begin{equation}
\boldsymbol{p}\circledast\boldsymbol{q}=F_{Q}^{-1}\left((F_{Q}\boldsymbol{p})\odot(F_{Q}\boldsymbol{q})\right),\label{eq:CCT}
\end{equation}
where $\boldsymbol{u}\odot\boldsymbol{v}\in\mathbb{C}^{Q}$ is obtained
by componentwise multiplication of $\boldsymbol{u}\in\mathbb{C}^{Q}$
and $\boldsymbol{v}\in\mathbb{C}^{Q}$.

The CCT provides efficient computation of a circular convolution using
the FFT and the IFFT. We notice that the summation on the right-hand
side of~\eqref{eq:CZT-CC} can obtained as part of a computation
of the form~\eqref{eq:CC}. To be precise, let $Q=2P$ and $\boldsymbol{p}\in\mathbb{C}^{Q}$,
$\boldsymbol{q}\in\mathbb{C}^{Q}$ be given by
\begin{equation}
p_{i}=\left\{ \begin{array}{ll}
x_{i}e^{-\imath i\frac{\pi\alpha}{M}(i-2R)}, & \quad\text{if}\quad0\le i<m,\\
0, & \quad\text{if}\quad m\le i<Q,
\end{array}\right.\label{eq:CZT-p}
\end{equation}
\begin{equation}
q_{i}=\left\{ \begin{array}{ll}
e^{\imath\frac{\pi\alpha}{M}i^{2}}, & \quad\text{if}\quad0\le i<P,\\
e^{\imath\frac{\pi\alpha}{M}(Q-i)^{2}}, & \quad\text{if}\quad P\le i<Q.
\end{array}\right.\label{eq:CZT-q}
\end{equation}
Then, we have
\begin{equation}
\hat{x}_{I}^{\alpha,M,P,R}=e^{-\imath\frac{\pi\alpha}{M}I^{2}}(\boldsymbol{p}\circledast\boldsymbol{q})_{I},\quad I\in\{0,1,\dots,P-1\}.\label{eq:CZT-equiv}
\end{equation}

Now we give an efficient algorithm for the computation of the CZT.
Since CZTs with fixed parameters $m$, $\alpha$, $M$, $P$, and
$R$ are often computed for several inputs, it is convenient to reuse
all data-independent computations. This is why we separate initialization,
which is performed by Algorithm~\ref{algo:CZT-init}, from the data-dependent
calculations executed by Algorithm~\ref{algo:CZT}, which makes use
of the parameters initialized by Algorithm~\ref{algo:CZT-init} while
computing the CZT itself.

\begin{algorithm}[H]
\caption{$\text{CZT-I}(m,\alpha,M,P,R)$}

\label{algo:CZT-init}

\normalsize{

\begin{algorithmic}[1]

\STATE{\textbf{for} \textbf{$i\in\{0,1,\dots,m-1\}$} \textbf{do}\label{step:CZT_coeffs_init}}

\STATE{$\quad\quad$$r_{i}\leftarrow e^{-\imath i\frac{\pi\alpha}{M}(i-2R)}$\label{step:CZT_coeffs_end}}

\STATE{\textbf{for} \textbf{$i\in\{0,1,\dots,2P-1\}$} \textbf{do}\label{step:CZT_initqhat_start}}

\STATE{$\quad\quad$$q_{i}\leftarrow\left\{ \begin{array}{ll}
e^{\imath\frac{\pi\alpha}{M}i^{2}}, & \quad\text{if}\quad0\le i<P,\\
e^{\imath\frac{\pi\alpha}{M}(2P-i)^{2}}, & \quad\text{if}\quad P\le i<2P.
\end{array}\right.$}

\STATE{$\hat{\boldsymbol{q}}\leftarrow\text{FFT}(\boldsymbol{q},2P,2P)$\label{step:CZT_initqhat_finish}}

\STATE{\textbf{for} \textbf{$I\in\{0,1,\dots,P-1\}$} \textbf{do}}

\STATE{$\quad\quad$$s_{I}\leftarrow e^{-\imath\frac{\pi\alpha}{M}I^{2}}$}

\STATE{$\textbf{return}$ $(\hat{\boldsymbol{q}},\boldsymbol{r},\boldsymbol{s})$
}

\end{algorithmic}}
\end{algorithm}

\begin{algorithm}[H]
\caption{$\text{CZT}(\boldsymbol{x},m,P,\hat{\boldsymbol{q}},\boldsymbol{r},\boldsymbol{s})$}

\label{algo:CZT}

\normalsize{

\begin{algorithmic}[1]

\STATE{\textbf{for} \textbf{$i\in\{0,1,\dots,m-1\}$} \textbf{do}\label{step:CZT_firstfor}}

\STATE{$\quad\quad$$p_{i}\leftarrow x_{i}r_{i}$\label{step:CZT_firstfor_end}}

\STATE{$\hat{\boldsymbol{p}}\leftarrow\text{FFT}(\boldsymbol{p},m,2P)$\label{step:CZT-dftp_p}}

\STATE{\textbf{for} \textbf{$I\in\{0,1,\dots,2P-1\}$} \textbf{do}\label{step:CZT-conv_start}}

\STATE{$\quad\quad$$\hat{p_{I}}\leftarrow\hat{p_{I}}\hat{q}_{I}$\label{step:CZT-second-for-finish}}

\STATE{$\boldsymbol{y}\leftarrow IFFT(\hat{\boldsymbol{p}},P,2P)$\label{step:CZT-conv_finish}}

\STATE{\textbf{for} \textbf{$I\in\{0,1,\dots,P-1\}$} \textbf{do}\label{step:CZT-final-start}}

\STATE{$\quad\quad$$y_{I}\leftarrow y_{I}s_{I}$\label{step:CZT-final-finish}}

\STATE{$\textbf{return}$ $(\boldsymbol{y})$}

\end{algorithmic}}
\end{algorithm}

Finally, let us relate Algorithm~\ref{algo:CZT} with the desired
computation~\eqref{eq:CZT-def}. We assume that we have computed
$(\hat{\boldsymbol{q}},\boldsymbol{r},\boldsymbol{s})\leftarrow\text{CZT-I}(m,\alpha,M,P,R)$
and then analyze the output of $(\boldsymbol{y})\leftarrow\text{CZT}(\boldsymbol{x},m,P,\hat{\boldsymbol{q}},\boldsymbol{r},\boldsymbol{s})$.
First we notice that after Steps~\ref{step:CZT_firstfor}-\ref{step:CZT_firstfor_end}
of Algorithm~\ref{algo:CZT}, because $\boldsymbol{r}$ was computed
in Steps~\ref{step:CZT_coeffs_init}-\ref{step:CZT_coeffs_end} of
Algorithm~\ref{algo:CZT-init}, we have $p_{i}=x_{i}e^{-\imath i\frac{\pi\alpha}{M}(i-2R)}$
for $i\in\{0,1,\dots,m-1\}$. Then, after Step~\ref{step:CZT-dftp_p}
of Algorithm~\ref{algo:CZT}, $\hat{\boldsymbol{p}}$ will hold exactly
the 1D DFT of the sequence defined in~\eqref{eq:CZT-p}. Because
$\hat{\boldsymbol{q}}$ was computed in Steps~\ref{step:CZT_initqhat_start}-\ref{step:CZT_initqhat_finish}
of Algorithm~\ref{algo:CZT-init}, we know that it holds the 1D DFT
of the sequence defined in~\eqref{eq:CZT-q}. Thus, due to the CCT~\eqref{eq:CCT},
we know that Steps~\ref{step:CZT-conv_start}-\ref{step:CZT-conv_finish}
of Algorithm~\ref{algo:CZT} give $y_{I}=(\boldsymbol{p}\circledast\boldsymbol{q})_{I}$
for $I\in\{0,1,\dots,P-1\}$. Steps~\ref{step:CZT-final-start}-\ref{step:CZT-final-finish}
of Algorithm~\ref{algo:CZT} complete the computation of~\eqref{eq:CZT-def}
through the equivalent formulation~\eqref{eq:CZT-equiv} with the
use of~\eqref{eq:CCT}.

\section{Relevance to Practical MRI in Medicine\label{sub:Motivate}}

Our desire to pursue the concept of GALFD for practical imaging applications
is driven by the recent advances in the use of radial imaging approaches
for MRI, with various sorts of polar-type sampling of the Fourier
space. Radial imaging has shown advantages for cardiac \cite{Feng2015}
and abdominal \cite{fgb14} imaging, as well as short-T2 \cite{Rahmer2006}
and sodium imaging \cite{Nielles-Vallespin2007}. While polar sampling
with projection-reconstruction was one of the original methods used
for MRI, it had been mostly abandoned for many years in favor of the
mathematically simpler and often more robust imaging methods that
use direct Cartesian grid sampling of the Fourier space; see Figure
\ref{fig:domains}. Linogram (\textquotedblleft concentric squares\textquotedblright )
data sampling approaches to radial imaging (as in Figure \ref{fig:GALFD},
left) were demonstrated early on to have computational advantages
over the more common \textquotedblleft concentric circles\textquotedblright{}
polar sampling schemes (as in Figure \ref{fig:domains}, right). This
is due to linogram sampling enabling the use of direct Fourier reconstruction
methods without the need for the time-consuming and resolution-degrading
intermediate interpolation steps used with reconstruction from samples
on concentric circles. However, this advantage was not compelling
enough to lead to their widespread adoption at that time, for the
relatively limited size data sets that were being acquired then. The
principal advantage of radial data sampling approaches for magnetic
resonance imaging is that it permits the use of very short delays
before acquiring the data, which is important for imaging of short-T2
species of signal sources \cite{Rahmer2006}, such as hydrogen in
large molecules (such as collagen or myelin) or sodium \cite{Nielles-Vallespin2007}.
In this case, the computational advantages of linogram imaging approaches
over conventional polar sampling methods become particularly important
for acquisition and reconstruction of fully 3D radial data sets. The
limitations of conventional polar radial image reconstruction methods
has led to the common use of \textquotedblleft stack-of-stars\textquotedblright{}
data acquisition approaches (\textquotedblleft concentric cylinders\textquotedblright )
for 3D radial imaging, with their longer associated data acquisition
delay times, rather than fully 3D radial sampling; using linogram
3D approaches (\textquotedblleft concentric cubes\textquotedblright{}
\cite{hra92}) thus has the potential for significantly improved (shorter)
delays in the data acquisition times, as compared to the conventional
approach.

Another factor leading to the recently renewed interest in the radial
sampling MRI methods has been their use for imaging of dynamic processes,
such as contrast enhancement dynamics \cite{fgb14}, or of moving
objects, such as the heart. In this case, we typically use golden
angle incrementation of the orientation of the consecutively acquired
radial sample sets, which permits continuous data acquisition with
the ability for flexible rebinning of the data for use in the image
reconstruction, with approximately uniform angular distributions of
the samples in each such set \cite{Feng2015} . However, this still
requires the time-consuming interpolation steps for reconstruction,
with its associated disadvantages for conventional 2D imaging, and
it has again led to the conventional use of the \textquotedblleft stack-of-stars\textquotedblright{}
approach for 3D imaging with golden angle increments, rather than
fully 3D radial imaging acquisitions. The use of modified linogram
approaches with golden angle increments, as demonstrated in the present
work, allows for the use of direct Fourier approaches for at least
one dimension of the image reconstruction, rather than needing to
rely on intermediate interpolations for all dimensions, with the associated
advantages of computational speed and precision, and thus can potentially
enable improved performance for MRI in these kinds of applications.
This is likely to be particularly useful for 3D imaging.

\section{Iterative Algorithms\label{sec:Iterative-Algorithms}}

Several practical circumstances require data acquisition in CT or
MRI to be done in a less than ideal manner leading to increased noise
and/or small number of data samples. Examples include imaging of moving
objects, requirements of radiation exposure reduction in CT, sodium
imaging in MRI, etc.

In such cases, methods that provide the best reconstructions usually
work by successively refining an initial estimate of the desired image,
that is, they are iterative. A reason why such methods are powerful
is their capability of solving optimization models of image reconstruction,
but quite a few efficacious iterative methods do not rely on an optimization
model. Many iterative methods share a common structure that we now
discuss for the problem of recovering an $\boldsymbol{x}\in\mathbb{C}^{mn}$
in \eqref{eq:DSFT_def} from experimentally obtained vector $\boldsymbol{y}$
of the values $\mathcal{D}[\boldsymbol{x}](\xi,\upsilon)$ of the
DTFT of $\boldsymbol{x}$ for all $(\xi,\upsilon)\in\mathfrak{D}$,
with $\mathfrak{D}$ a GALFD.

Such an iterative method starts with some $\boldsymbol{x}^{0}\in\mathbb{C}^{mn}$
and proceeds by successive application of an algorithmic operator
$\mathcal{A}$ to the current iterate $\boldsymbol{x}^{k}$ to obtain
the next one: $\boldsymbol{x}^{k+1}=\mathcal{A}(\boldsymbol{x}^{k})$,
$k=0,1,2,\dots$. It is in the computation of $\mathcal{A}(\boldsymbol{x}^{k})$
that GALE can be efficacious. To demonstrate this, we refer to \cite[Section 12.2]{her09}
that discusses in detail a particular iterative method of CT reconstruction.
For example, Equation (12.25) there can be rewritten in the notation
of our current paper as
\begin{equation}
\boldsymbol{x}^{k+1}=\mathbf{x}^{k}+\lambda^{k}\left(r^{2}\mathcal{O}^{*}\left(\boldsymbol{y}-\mathcal{O}\boldsymbol{x}^{k}\right)+\left(\mu_{X}-\boldsymbol{x}^{k}\right)\right),\label{iteration}
\end{equation}
where $\mathcal{O}$ is the operator calculated by Algorithm 2 (GALE)
and $\mathcal{O}^{*}$ is its adjoint, $\lambda^{k}$ is a real number
(the relaxation parameter in the $k$th iteration), $r$ is a real
number (the signal-to-noise ratio in the data collection) and $\mu_{X}$
is the expected prior value of $\boldsymbol{x}$. If data were acquired
in a GALFD, then GALE could be used for efficient computations of
$\mathcal{O}\boldsymbol{x}^{k}$ and of $\mathcal{O}^{*}\left(\boldsymbol{y}-\mathcal{O}\boldsymbol{x}^{k}\right)$,
which are often the most time consuming operations of each iteration.
This would lead to substantial computational savings that could bring
such iterative reconstruction techniques closer to clinical practice.

\section*{References}

\bibliographystyle{iopart-num}
\bibliography{refs_Elias_2019-08-06}

\end{document}